%% file: arxiv-risk-averse-game.tex
\documentclass[10pt,english]{article}
\usepackage{geometry}
\usepackage[T1]{fontenc}
\usepackage[latin9]{inputenc}
\usepackage{bm}
\usepackage{amsmath}
\usepackage{amssymb} 
\usepackage[unicode=true,
 bookmarks=false,pagebackref=true,
 breaklinks=false,pdfborder={0 0 1},colorlinks=false]{hyperref}
\hypersetup{colorlinks,citecolor=blue,filecolor=blue,linkcolor=blue,urlcolor=blue}
\usepackage{xcolor,colortbl}
\definecolor{Gray}{gray}{0.85}
\usepackage{enumitem}
\makeatletter
\usepackage{amsthm}
\usepackage{cite}  
\usepackage{comment}
\usepackage{booktabs,mathtools}
\usepackage{graphicx}
\usepackage[linesnumbered,ruled,vlined]{algorithm2e}
\usepackage{algorithmic}

\usepackage{multirow}
\usepackage{dsfont}
\usepackage{color}
\usepackage{float}
\usepackage{nicefrac}       
\usepackage{microtype}      
\usepackage{xcolor} 
\usepackage{amsmath}
\usepackage{amsthm}
\usepackage{amsfonts}
\usepackage{amssymb}
\usepackage{bm}
\usepackage[capitalize]{cleveref}
\usepackage{mathtools}
\usepackage{dsfont}

\newcommand{\mc}[1]{\mathcal{#1}}
\newcommand{\mb}[1]{\mathbb{#1}}

\usepackage[linesnumbered,ruled,vlined]{algorithm2e}
\usepackage{algorithmic}

\usepackage{wrapfig}

\definecolor{lxs}{RGB}{138,43,226}

\usepackage{multirow}
\input{macro}

\allowdisplaybreaks

\title{Tractable Equilibrium Computation in Markov Games through Risk Aversion}

 \author{
	Eric Mazumdar\thanks{Author list is in alphabetical order.}~\thanks{Department of Computing Mathematical Sciences and Department of Economics, California Institute of Technology, CA 91125, USA.}\\
 	Caltech 
	\and
	Kishan Panaganti\footnotemark[1]~\thanks{Department of Computing Mathematical Sciences, California Institute of Technology, CA 91125, USA.} \\
	Caltech
 	\and
    Laixi Shi\footnotemark[1]~\footnotemark[3] \\ 	 
	Caltech
 	} 

\date{\today}

\begin{document}
\theoremstyle{plain} \newtheorem{lemma}{\textbf{Lemma}}
\newtheorem{proposition}{\textbf{Proposition}}
\newtheorem{theorem}{\textbf{Theorem}}
\newtheorem{assumption}{Assumption}
\newtheorem{corollary}{Corollary}[theorem] 
\newtheorem{definition}{Definition}
\newtheorem{example}{Example}

\theoremstyle{remark}\newtheorem{remark}{\textbf{Remark}}

\maketitle
 
  \sloppy

\begin{abstract}
A significant roadblock to the development of principled multi-agent reinforcement learning is the fact that desired solution concepts like Nash equilibria may be intractable to compute. To overcome this obstacle, we take inspiration from behavioral economics and show that---by imbuing agents with important features of human decision-making like risk aversion and bounded rationality---a class of risk-averse quantal response equilibria (RQE) become tractable to compute in all $n$-player matrix and finite-horizon Markov games.  In particular, we show that they emerge as the endpoint of no-regret learning in suitably adjusted versions of the games. Crucially, the class of computationally tractable RQE is independent of the underlying game structure and only depends on agents' degree of risk-aversion and bounded rationality. To validate the richness of this class of solution concepts we show that it captures peoples' patterns of play in a number of 2-player matrix games previously studied in experimental economics. Furthermore, we give a first analysis of the sample complexity of computing these equilibria in finite-horizon Markov games when one has access to a generative model and validate our findings on a simple multi-agent reinforcement learning benchmark.
\end{abstract}

\noindent \textbf{Keywords:} behavioral economics, risk-aversion, multi-agent reinforcement learning, quantal response, bounded rationality.

\allowdisplaybreaks

\setcounter{tocdepth}{2}

\section{Introduction} \label{sec:intro}
Machine learning algorithms are increasingly being deployed in dynamic environments in which they interact with other agents like people or other algorithms. Often, these agents have their own goals which may not be aligned with those of the algorithm---making the interactions \emph{strategic}. Such situations are naturally modeled as \emph{games} between rational agents and their prevalence in a wide range of application areas has driven a surge in research interest in learning in games~\cite{bianchi_prediction} and multi-agent reinforcement learning~\cite{zhang2021multi} in recent years. Indeed, real-world applications of these problems range from the training of large language models~\cite{munos2023nash} to the decentralized control of the smart grid \cite{mohsenian2010autonomous}, autonomous driving \cite{kannan2017fairness}, and financial trading \cite{wellman1998market} among many others.

When viewed through the lens of game-theory, many of these problems can be cast as problems of \emph{equilibrium computation} under varying information structures, where the equilibrium represents a stable outcome for rational agents. The most common equilibrium concept is that of a Nash equilibrium (NE) \cite{nash1950non}: a solution under which no rational agent has an incentive to unilaterally seek to improve their outcome. Despite its popularity as a solution concept, computing a NE outside of highly structured games is known to be computationally intractable \cite{daskalakis2013complexity} even for two-player matrix games. 

While relaxations of NE like (coarse) correlated equilibria (CCE) are known to be more tractable to compute, they also have their limitations. Indeed while CCE can be computed through the use of no-regret learning algorithms~\cite{bianchi_prediction}, the set of CCE can be large (introducing an additional problem of equilibrium selection) and may include strictly dominated strategies~\cite{no_regret_dominated} which means that they cannot necessarily be rationalized by individual agents~\cite{DEKEL1990243}. Furthermore, a refinement of CCE---stationary CCE--- are also known to be PPAD-hard to compute in dynamic general-sum games~\cite{MarkovComplex}.

Beyond these hardness results, solution concepts like NE and CCE also fail to be predictive of what strategies people play in games~\cite{mckelvey1995quantal,prediction1}, with people being observed to be imperfect optimizers~\cite{StochasticGameTheory,noisy_decisions} and risk-averse~\cite{goeree2003risk} when confronted with game theoretic scenarios. This aligns with celebrated work in behavioral economics and mathematical psychology which has repeatedly shown that dominant features of human decision-making are  a failure to perfectly optimize~\cite{Luce59} and risk-aversion~\cite{Kaheman,Tversky1992}. 

The first observation is often referred to as \emph{bounded rationality} which posits that individuals are naturally prone to making mistakes and often fail to be perfectly optimal~\cite{Luce59}. This is often captured in games through the idea of a \emph{quantal response equilibrium} (QRE) \cite{mckelvey1995quantal}. 

The second observation can be attributed to the fact that players typically face uncertainty and risk in their decisions, stemming from environmental uncertainties like unknown future events, noise, or even other players' lack of optimality. The presence of these uncertainties can lead people to prefer risk-averse strategies, i.e., strategies which give more certain outcomes at the cost of lower expected returns \cite{gollier2001economics}. Interestingly, there is evidence that neither of these properties alone can account for people's patterns of play observed in controlled experiments~\cite{goeree2003risk,GoereeAuction}, and that models of decision-making that incorporate \emph{both} of these features have the best predictive power~\cite{goeree2003risk}.

Motivated by these findings, in this paper we formulate and study games in which agents are risk-averse and have bounded rationality and study the computational tractability of the natural equilibrium concept in this class of games: a risk-averse quantal response equilibrium (RQE). At first glance, the introduction of these features of human decision-making into games can break existing game theoretic structures. However, by relying on dual formulations of risk we show how, for a large range of degrees of risk-aversion and bounded rationality, RQE are computationally tractable in \emph{all} n-player matrix games. Crucially, the class of computationally tractable RQE is independent of the original game structure and depends only on players' degrees of risk aversion and bounded rationality. We then extend the result to finite-horizon Markov games in which agents interact in a Markov decision process (MDP) and provide a finite-sample complexity guarantee on how many samples one would need to approximate a RQE in Markov games.

\begin{figure}[t!]
		\centering
	\includegraphics[width=\linewidth]{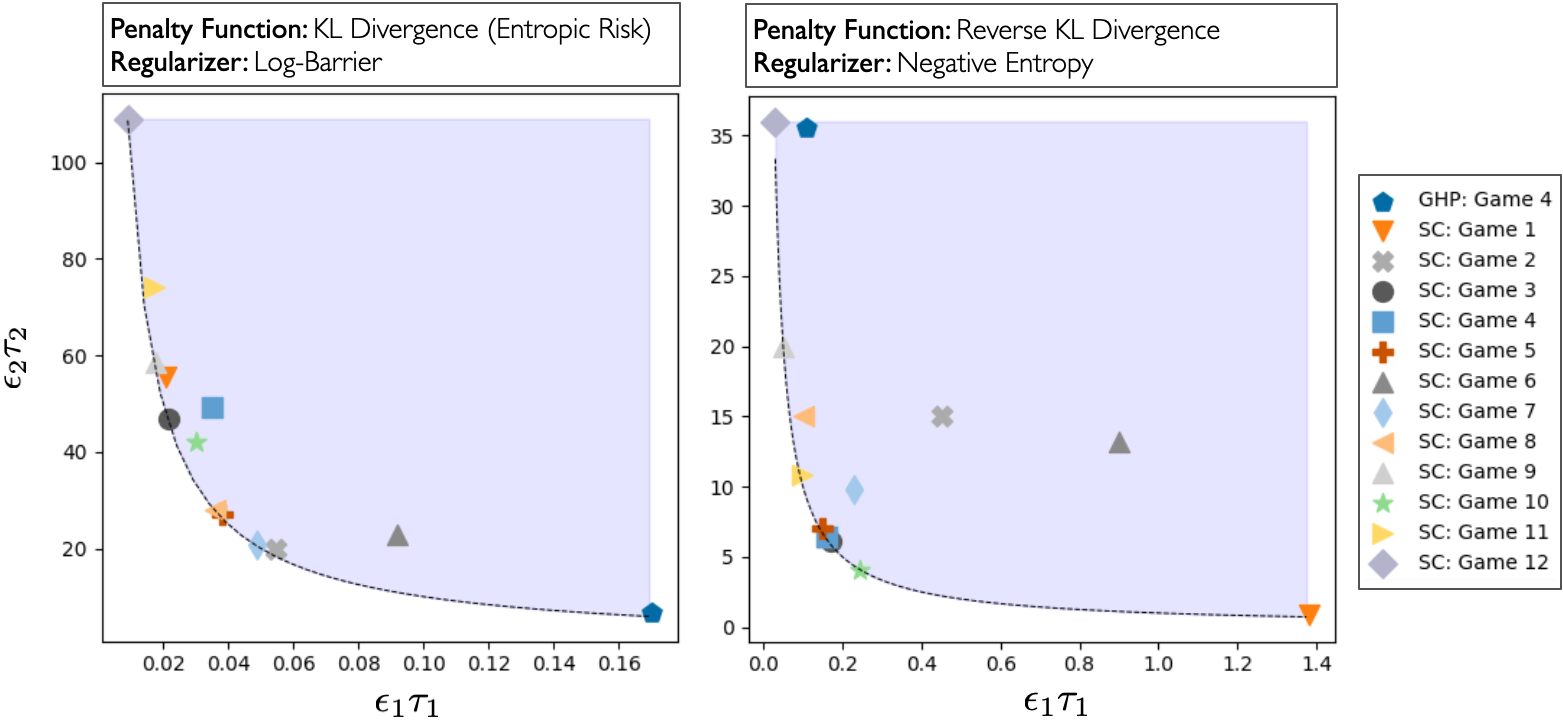}
    \vspace{-0.3cm}
	\caption{ The shaded blue region depicts the regime of risk-aversion and bounded rationality preferences that allow for computationally tractable RQE in all $2$-player games as shown in \cref{thm:2player}.
    The markers \textit{GHP: Game 4} \cite{goeree2003risk} and \textit{SC: Game 1-12} \cite{selten2008stationary} represent the necessary parameter values required to recreate the average strategy played by people in various $2$-player games in observational data up to $1\%$ accuracy.}
	\label{fig:kl-rkl-plot}
\end{figure}

These results suggest that designing multi-agent reinforcement learning (MARL) agents with realistic features of human decision-making can allow for a computationally tractable equilibrium concept in all finite action games and finite horizon Markov games. To emphasize the practical relevance of this theoretical result, we show how that the regime for which RQE are computationally tractable captures real-world data from behavioral economics on people's observed pattern of play in $13$ different games. This is illustrated in \cref{fig:kl-rkl-plot}, where the blue region represents the set of RQE that are computationally tractable which is a function of agents'  degree of risk-aversion ($\tau_1/\tau_2$), and level of bounded rationality ($\epsilon_1/\epsilon_2$).


\paragraph{{\bf Contributions:}} To begin, in Section~\ref{sec:Matrix} we focus on multi-player general-sum matrix games. We first introduce risk-aversion into matrix games by allowing agents to be risk-averse to the randomness introduced into the game by their opponents. To incorporate risk-aversion we make use of a general class of convex risk measures and present two forms of risk aversion: \emph{aggregate risk-aversion} and \emph{action-dependent risk aversion}---with the latter yielding less conservative solutions at the cost of a more complex formulation and analysis. Given this setup, we show that risk-averse Nash equilibria must exist in \emph{both} formulations in all games and for any convex measures of risk. 

We then introduce bounded rationality into the games by assuming that agents optimize over quantal responses instead of the entire probability simplex. In Section~\ref{sec:computation_matrix} we show that computing the QRE of both types of risk-averse matrix games---i,e., the RQE described above---can be done in polynomial time as long as the level of risk-aversion and class of quantal response functions satisfies certain conditions.  Importantly, these conditions are independent of the underlying game and only depend on the class of risk metrics and quantal responses under consideration. Thus, RQE not only incorporates features of human decision-making but are also more amenable to computation than QRE or NE in matrix games.

In Section~\ref{sec:Markov}, we extend the study to multi-player general-sum Markov games. In this setting we consider the case in which players are risk-averse not only to the randomness introduced by their opponents, but also to the randomness in the environment. We again show that RQE can be computed efficiently when the environment is known and one has full information over the game. In addition, we consider the multi-agent reinforcement learning setting in which the environment is unknown and one needs to learn by sampling. We provide a first algorithm and a finite-sample guarantee for learning an approximate RQE in such MARL problems. Finally, we conduct experiments of the dynamic games in gridworlds to illustrate the effect of risk-aversion in sequential games and validate our findings.

\subsection{Related Works}

Before presenting our results, we put our work into context, and discuss related works here.

\paragraph{Computational tractability of game theoretic solution concepts.}
This work proposes a new solution concept for game theoretic settings that is computationally tractable, yet retains many of the desirable properties of classical equilibrium concepts. This general question emerged from the finding that computing a Nash equilibrium---perhaps the most natural solution concept for a game between rational self-interested agents---is PPAD-hard~\cite{daskalakis2013complexity}, even for two-player general-sum matrix games. Despite this negative result, a large amount of subsequent work has focused on understanding the classes of games in which one can compute, approximate, or learn Nash equilibria efficiently. This is often done by assuming additional structure on the players' utilities and their relationships to one another, with large classes of games being zero-sum or competitive games, zero-sum polymatrix games~\cite{eq_collapse1,eq_collapse_Markov}, monotone games~\cite{golowich2020tight}, smooth games~\cite{roughgarden2015intrinsic}, or socially concave games~\cite{even2009convergence}.

In games without such structure however, the natural targets for computation and learning became correlated~\cite{CCE1} and coarse correlated equilibria ~\cite{AUMANN197467,CE2}(CE and CCE respectively), both of which can be shown to emerge as the endpoint of no-regret learning and are thus considered to be computationally tractable targets for the design of learning algorithms. Despite this desirable property, the two concepts have significant drawbacks. Indeed both CE and CCE require some form of coordination between players to implement, introduce a highly nontrivial equilibrium selection problem~\cite{bianchi_prediction}, and may have support on dominated strategies~\cite{no_regret_dominated}. Furthermore, in the dynamic game context of Markov games, stationary  CE and CCE are also computationally intractable to compute~\cite{MarkovComplex}.

More recently, a new equilibrium concept---a smoothed Nash equilibrium--- has been proposed as an alternative to these other equilibrium concepts~\cite{daskalakis2023smooth} and motivated by similar considerations of individual and independent rationalizability and computational tractability. By applying ideas from smoothed analysis to the problem of computing Nash equilibria the authors show that one can efficiently find approximate classes of smoothed Nash equilibria---though to the best of our knowledge this cannot be done in a decentralized way. 

Our approach is orthogonal and is rooted in giving MARL agents a foundation rooted in behavioral economics by imbuing them with a realistic feature of human decision-making: risk-aversion. The question of computational tractability of risk-averse Nash equilibria has been analyzed in~\cite{fiat2010players}. The work shows that if agents are risk-averse with respect to all the randomness in the game (including their own) then a risk-averse Nash equilibrium may not even exist in mixed strategies, and even understanding if such equilibria exist can be NP-complete. Our formulation overcomes this by incorporating risk-aversion in a different way. Indeed, we show that when agents are risk-averse \emph{only to the randomness introduced to their opponents (and the environment)} then the risk-averse Nash equilibria will \emph{always} exist. We note that such formulations of risk-aversion are common in the literature on risk-sensitive control~\cite{shenControl,BorkarRisk} and risk-sensitive reinforcement learning~\cite{shen2014risk} where agents are implicitly presumed to be risk-averse only to the randomness that is outside their control (i.e., the environment). Furthermore we show that introducing bounded rationality into the game allows a class of risk-averse quantal response equilibria (RQE) to be computationally tractable in \emph{all} finite action and finite-horizon Markov games.

\paragraph{Predictive power of equilibrium concepts.} Another driving force in moving beyond the Nash and correlated equilibrium concepts stems from their lack of predictive power in experimental settings (see e.g., ~\cite{Experiments1,Experiments2,Experiment3,prediction1,mckelvey1995quantal}). To address this, a line of work originating in economics seeks to understand the natural solution concepts in game where players have behaviorally plausible restrictions to their strategy spaces, and to study whether such equilibria were better predictors of human play than Nash or (coarse) correlated equilibria~\cite{goeree2003risk,GoereeAuction,Camerer}. The most common restriction is that players have \emph{bounded rationality},---i.e., they may fail to perfectly optimize---a model with roots in mathematical psychology~\cite{Luce59}. Under this restriction, a natural equilibrium concept that emerged was that of a quantal response equilibrium (QRE) which induces bounded rationality by either assuming that the players are rational in a stochastically perturbed version of the game or equivalently that they optimize a regularized version of their utility~\cite{mckelvey1995quantal,mckelvey1998quantal,quantal2}. Beyond their use as a better model for human decision-making in games, QRE have also increasingly been adopted as a solution concept in multi-agent reinforcement learning and learning in games~\cite{quantal1,quantal2,cen2021fast,quantal3,evans2024learning,jacob2022modeling} due to their links with KL and entropy regularized reinforcement learning. Despite these developments QRE are not computable in all games. Indeed the class of QRE or equivalently the level of bounded rationality needed for computational tractability depends on the underlying game structure which may not be known a priori. In contrast we show that the addition of risk aversion allows for the \emph{same} class of quantal response equilibria to be computationally tractable to compute in all finite action games and finite-horizon Markov games. Furthermore we show that this class of risk-averse QRE is nontrivial and can capture human data better than risk-neutral QRE---a finding which is in line with findings in behaviorial economics~\cite{goeree2003risk,GoereeAuction}.

\paragraph{Risk-averse and robust multi-agent reinforcement learning.}
Our work builds on and provides an additional justification for risk-sensitive (multi-agent) reinforcement learning. This line of work has roots going back to seminal work by Jacobson on risk-sensitive control~\cite{Jacobsen}, and more recently in risk-sensitive reinforcement learning~\cite{shen2014risk}. In these works, the aim is to find a controller or policy for a system that accounts for stochasticity or uncertainty in the environment or system in a more nuanced way than risk-neutral approaches like optimal control or reinforcement learning~\cite{BorkarRisk}. Due to classic duality results (see e.g., ~\cite{zhang2024soft}) this line of work is closely related to the literature on robust control and distributionally robust reinforcement learning~\cite{iyengar2005robust,panaganti2020robust,shi2022distributionally,xu-panaganti-2023samplecomplexity} which seeks to find solutions that are robust to worst case environmental disturbances. 

Our work rigorously extends these formulations to the multi-agent regime though it is not the first to consider risk-aversion in MARL. Indeed, risk-sensitive MARL has been the focus of several recent works (e.g., \cite{yekkehkhany2020risk,wang2024learning,gao2021robust,slumbers2023game}). Several provide rigorous definitions of risk-averse equilibria and some guarantees on their computation by assuming structure on the risk-adjusted game. Oftentimes this is done by assuming that the risk-adjusted game is itself zero-sum~\cite{yekkehkhany2020risk}, monotone~\cite{wang2024learning}, or that it satisfies other strong conditions~\cite{gao2021robust}. Other works are more empirical in nature~\cite{ganesh2019reinforcement,eriksson2022risk,zhang2021mean,qiu2021rmix,shen2023riskq,slumbers2023game}, showing the promise of risk-averse algorithms for MARL. 

One last closely related line of work is the emerging literature on robust multi-agent reinforcement learning~\cite{zhang2020robust,he2023robust,shi2024sample,blanchet2024double}. Once again due to duality arguments, these works can be seen as tackling a similar problem to the risk-averse MARL problem. The focus of these previous works, however, is on robustness in the face of only environmental uncertainties (and not opponent strategies), and questions of existence and computational tractability are either assumed away or the focus is on extensions of correlated equilibrium concepts. A recent related work in this literature analyzed the  computational tractability of robust Nash equilibria in Markov games, but only provided strong guarantees on the zero-sum regime, showing that computing such equilibria in general is PPAD-hard~\cite{mcmahan2024roping}. 

To the best of our knowledge, no previous work in either of these literatures highlights the broad benefits afforded by risk-aversion in MARL in terms of computational tractability of equilibria. In our work we show that risk-aversion (and by extension distributional robustness), when combined with bounded rationality yields a computationally tractable class of individually rationalizable equilibria in \emph{all} finite-horizon $n$-player Markov games. Furthermore we show that these equilibria can be computed using no-regret learning algorithms.

\paragraph{Notations:}
We use $\Delta_{n}$ to denote probability simplex of size $n$. In addition, we denote $[N] = \{ 1,2,\cdots, N\}$ for any positive integer $N >0$. We denote $x = \big[x(s,a)\big]_{(s,a)\in\cS\times\cA}\in \mathbb{R}^{SA}$ (resp.~$x = \big[x(s)\big]_{s\in\cS}\in \mathbb{R}^{S}$) as any vector that constitutes certain values for each state-action pair (resp.~state).

\section{Matrix Games}\label{sec:Matrix}

To begin, we consider $n$-player general-sum finite-action games. In these games, each player $i$ has access to a (finite) action set $\mathcal{A}_i$ with $A_i=|\mathcal{A}_i|$ pure strategies. For each tuple of joint strategies $a=(a_1,...,a_n)\in \mathcal{A} \defn \prod_{i=1}^n\mathcal{A}_i$ each player has an associated reward or utility $R_i(a)$. As is common in the study of these games, we consider the case where players play over mixed strategies and seek to maximize their expected utility. Player $i$'s expected utility in this case can be written as a function $U_i:\mathcal{P}\rightarrow \mathbb{R}$ which can be written as:
\begin{align}
\begin{split}\label{eq:exputility}
    U_i(\pi_1,...,\pi_n)= \mathbb{E}_{a\sim \pi}[R_{i}(a)]
    \end{split}
\end{align}
where  $\pi=(\pi_1,...,\pi_n) \in \mathcal{P}=\prod_{i=1}^n \Delta_{A_i}$ is the joint strategy of all the players. For ease of exposition we often write the utility as $U_i(\pi_i,\pi_{-i})$ where $\pi_{-i}=(\pi_1,...,\pi_{i-1},\pi_{i+1},...,\pi_n)$ represents the joint strategies of all players \emph{other} than $i$.

We note that the expected utility in finite-action games has a multi-linear form and that $U_i(\pi_i,\pi_{-i})$ is convex (actually linear) in $\pi_i$ for fixed $\pi_{-i}$. To see this, one can observe that in two player games, these utility functions take the simpler, bilinear forms:
\begin{align}
\begin{split}
    U_1(\pi_1,\pi_2)= \mathbb{E}_{\substack{a_1\sim \pi_1\\ a_2\sim x_2}}[R_{1}(a_1,a_2)]= \pi_1^T R_1 \pi_2\\
    U_2(\pi_1,\pi_2)= \mathbb{E}_{\substack{a_1\sim \pi_1\\ a_2\sim \pi_2}}[R_{2}(a_1,a_2)]= \pi_2^T R_2 \pi_1,
    \end{split}
\end{align}
where $R_1 \in \mb{R}^{A_1 \times A_2} $ and $R_2 \in \mb{R}^{A_2 \times A_1}$ are payoff matrices representing the utility or reward associated to each joint pure strategy.

A natural outcome in this class of games is the notion of a (mixed) Nash equilibrium. 

\begin{definition}[Nash Equilibrium]
    A (mixed) Nash equilibrium of a $n$-player general-sum finite-action game is a joint strategy $\pi^*=(\pi_1^*,...,\pi_n^*)\in \mathcal{P}$ such that no player has any incentive to unilaterally deviate---i.e,. for all $i=1,...,n$:
    \begin{align*}
    U_i(\pi_i^*,\pi_{-i}^*)\ge U_i(\pi_i,\pi_{-i}^*) \quad \forall \pi_i \in \Delta_{A_i}
\end{align*}
\end{definition}

While such games are well known to admit at least one Nash equilibrium in mixed strategies, the Nash equilibrium may be intractable to compute. Furthermore, a preponderance of empirical evidence suggests that people do \emph{not} play their Nash strategies and in fact are risk-averse and potentially bounded rational in their decision-making--- inducing new forms of equilibria. 

In the next two sections, we introduce generalizations of the expected utility game that allow us to model agents as both risk-averse and imperfect optimizers---both of these being key attributes of human decision-making. We show that while the addition of risk can undo the existing structure in the game, the addition of bounded rationality in the form of quantal responses allow us to derive a class of computationally tractable equilibria in \emph{all} games that only depends on the degree of risk aversion and bounded rationality and not on the game structure. 

\subsection{Risk-Aversion in Matrix Games}

Given the setup of an expected utility game, we now allow agents to have risk preferences which we model through the use of a general class of convex risk metrics which have been the focus of a long line of study in mathematical finance and operations research~\cite{Risk_overview}. In this framing we move into a regime where agents seek to \emph{minimize} a measure of risk.

\begin{definition}[Convex Risk Measures]\label{def:convexmetric}
    Let $\mathcal{X}$ be the set of functions mapping from a space of outcomes $\Omega$ to $\mathbb{R}$. A convex measure of risk is a mapping $\rho:\mathcal{X}\rightarrow \mathbb{R}$ satisfying:
    \begin{enumerate}
        \item \emph{Monotonicity:} If $X\le Y$ almost surely, then $\rho(X)\ge \rho(Y)$.
        \item \emph{Translation Invariance:} If $m\in \mathbb{R}$ then $\rho(X+m)=\rho(X)-m$.
        \item \emph{Convexity:} For all $\lambda \in (0,1)$, $\rho(\lambda X+(1-\lambda)Y)\le \lambda\rho(X)+(1-\lambda)\rho(Y)$.
    \end{enumerate}
\end{definition}

Typically, and as we assume in the remainder of the paper, the set $\mathcal{X}$ is the set of measurable functions defined on a probability space $(\Omega,\mathcal{F},P)$. Under this assumption, the convex measures of risk allow us to generalize expectations to allow for agents to optimize for more conservative or worst-case outcomes. For ease of exposition, we write $\rho_\pi(X)$ to denote the distribution that the measure of risk is taken with respect to a reference distribution $\pi$.

In the case of matrix games, we consider the case where agents are risk averse with respect to the mixing or randomness introduced into the game by their opponents. In generalizing this to stochastic games in future sections we will also consider the case where agents are risk averse with respect to the underlying dynamics in the Markov game. Crucially, we assume that players are not risk averse to their own randomness.  We note that is a common approach taken in the literature on risk-sensitive and robust decision-making~\cite{shen2014risk} and it is necessary since if agents are risk-averse to their own randomness then an equilibrium may cease to exist~\cite{fiat2010players}. We refer to related works for more discussion.

We consider two generalizations of the original game that differ in how risk is incorporated into the problem:  \emph{aggregate risk aversion} and \emph{action-dependent risk aversion}. The first generalization allows for a comparatively simpler analysis but at the cost of a more conservative solution while the second is less conservative but requires a more complex analysis. Due to monotonicity and linearity of expectation both of these formulations can be seen as generalizations of risk-sensitive decision-making investigated in the literature on risk-sensitive control~\cite{BorkarRisk} and risk-sensitive RL~\cite{shen2014risk}. To see this reduction, one can simply take the other agents to be part of an unknown environment. 

\paragraph{Aggregate Risk Aversion:} To define a game with aggregate risk aversion, we transform the player's utilities in \eqref{eq:exputility} into costs $f_i$ which take the form
\begin{align}
\begin{split}\label{eq:riskmat}
    f_i(\pi_i,\pi_{-i})= \rho_{i,\pi_{-i}}(\mathbb{E}_{\pi_i}[R_{i}(a)])=\rho_{i,\pi_{-i}}\left( \sum_{a_i \in \mathcal{A}_i} \pi_i(a_i)R_i(a_i,a_{-i})\right)
    \end{split}
\end{align}
where $\rho_{i,\pi_{-i}}$ is used to denote the potentially different risk preference of agent $i$ which depends on the product distribution of opponents strategies $\pi_{-i}$. 

\paragraph{Action-dependent Risk Aversion:} To define a game with aggregate risk aversion, we transform the player's utilities in \eqref{eq:exputility} into costs $f_i$ which take the form
\begin{align}
\begin{split}\label{eq:action_dep}
    f_i(\pi_i,\pi_{-i})= \mathbb{E}_{\pi_i}\left[\rho_{i,\pi_{-i}}(R_{i}(a))\right]= \sum_{a_i \in \mathcal{A}_i}\pi_i(a_i)\rho_{i,\pi_{-i}}\left( R_i(a_i,a_{-i})\right)
    \end{split}
\end{align}
where again $\rho_{i,\pi_{-i}}$ is used to denote the potentially different risk preference of agent $i$ which depends on the product distribution of opponents strategies $\pi_{-i}$.

We remark that in both of these formulations, if $\rho_i(X)=\mathbb{E}[-X]$ for all players $i=1,...,n$ (which satisfies the requirements in Definition~\ref{def:convexmetric}) then the new formulation reduces to the original expected utility objective. Thus, both can be seen as generalizations of the classic setup described in the previous section.

While---at first glance--- the modified games look significantly more complex than the previous expected utility maximization setup (cf.~\eqref{eq:exputility}), we can rely on a particularly powerful property of convex measures of risk to simplify and expose some structure in this class of problems. In particular, we make use of the \emph{dual} representation theorem for convex risk measures. 

\begin{theorem}[Dual Representation Theorem for Convex Risk Measures~\cite{Risk_overview}]
Suppose that the set $\mathcal{X}$ is the set of  functions mapping from a finite set $\Omega$ to $\mathbb{R}$. Then a mapping $\rho:\mathcal{X}\rightarrow \mathbb{R}$ is a convex risk measure (cf.~Definition~\ref{def:convexmetric}) if and only if there exists a penalty function $D:\Delta_\Omega\rightarrow (-\infty,\infty]$ such that:
$\rho(X)=\sup_{p\in \Delta_\Omega} E_p[-X]-D(p),$
where $\Delta_\Omega$ is the set of all probability measures on $\Omega$. Furthermore, the function $D(p)$ can be taken to be convex, lower-semi-continuous, and satisfy $D(p)>-\rho(0)$ for all $p\in \Delta_{\Omega}$. 
\end{theorem}

When the set $\mathcal{X}$ is again the set of  measurable functions defined on a probability space, one can choose the penalty function $D$ to represent a notion of of distance from the probability law or distribution of the random variable $\pi$. In such cases, the dual representation theorem takes the form:
\[\rho_\pi(X)=\sup_{p\in \Delta_\Omega} E_p[-X]-D(p,\pi),\]
where $D(p,\pi)$ is convex in $p$ for a fixed $\pi$. We also make a simplifying assumption that $D$ is continuous in both its arguments, which is satisfied by various widely-used risk measures. This general form allows us to draw connections with a large class of risk and robustness metrics that are based around $\phi$-divergences. We provide examples of common risk measures in \cref{tab:risk-measure}. 

\newcommand{\topsepremove}{\aboverulesep = 0mm \belowrulesep = 0mm} \topsepremove

\begin{table}[t!]
	\begin{center}
\begin{tabular}{c|c}
\hline
\toprule
\multirow{2}{*}{Risk-measure }	&\multirow{2}{*}{Penalty function $D(p,q)$ }  \tabularnewline 
	&   \tabularnewline
\hline
\toprule
\multirow{2}{*}{Entropic Risk  \cite{ahmadi2012entropic} } &
 \multirow{2}{*}{Kullback-Leibler (KL): {$KL(p,q)=\sum_{i}p_{i}\log\big(\frac{p_{i}}{q_{i}}\big)$  \vphantom{$\frac{1^{7}}{1^{7^{7}}}$} }  \vphantom{$\frac{1^{7}}{1^{7^{7}}}$} } \tabularnewline
  &  \tabularnewline
\hline
\multirow{2}{*}{ \cite{Risk_overview} }  & 
 \multirow{2}{*}{Reverse KL (RKL): $\sum_{i}q_{i}\log\left(\frac{q_{i}}{p_{i}}\right)$  \vphantom{$\frac{1^{7}}{1^{7^{7}}}$}}   \tabularnewline
  &   \tabularnewline
  \hline
\multirow{2}{*}{ $\phi$-Entropic Risk  \cite{ahmadi2012entropic}} & 
 \multirow{2}{*}{ $\phi$-Divergence:  $\sum_{i}p_{i}\, \phi\left(\frac{p_{i}}{q_{i}}\right)$  \vphantom{$\frac{1^{7}}{1^{7^{7}}}$}}   \tabularnewline
  & \tabularnewline
  \hline
\multirow{2}{*}{ Utility-based shortfall   \cite{Risk_overview}}  & 
 \multirow{2}{*}{Utility-based Shortfall ($u$):  $\inf_{\lambda>0} \frac{1}{\lambda}[ r + \sum_{i}p_{i} u^\star\big(\lambda\frac{q_{i}}{p_{i}}\big) ]$  \vphantom{$\frac{1^{7}}{1^{7^{7}}}$}}  \tabularnewline
&    \tabularnewline
\toprule
\end{tabular}
	\end{center}
	\caption{We list several widely-used convex risk-measures with its penalty function $D(p,q)$. Here, $p,q$ are distributions of the same finite dimension, $\phi$ and $u$ are differentiable convex functions on $\mb{R}$ satisfying $\phi(1)=0$ and $\phi(t)=+\infty$ for $t<0$. The utility $u$ is equipped with risk tolerance level of $r$, and its penalty function depends on its convex conjugate $u^\star$. }
    \vspace{-0.5cm}
 \label{tab:risk-measure} 
\end{table}

Given this setup, we can now write the aggregate risk-adjusted game in the following form:
\begin{align}
\begin{split}\label{eq:riskmat2}
    f_i(\pi_i,\pi_{-i})= \sup_{p_i \in \mathcal{P}_{-i}} -\pi_i^T R_i p_i -D_i(p_i,\pi_{-i}) \\
    \end{split}
\end{align}
where $\mathcal{P}_{-i}=\mathcal{P}/\Delta_{A_i}\subset \mathbb{R}^{A_{-i}}$, $A_{-i}=\sum_{j\ne i} A_j$, and $R_i \in \mathbb{R}^{A_i\times A_{-i}}$ is player $i$'s payoff matrix. Similarly, the action-dependent risk-adjusted game takes the form:
\begin{align}
\begin{split}\label{eq:action_dep2}
    f_i(\pi_i,\pi_{-i})= \sum_{j \in \mathcal{A}_i} \pi_i(j)\left(\sup_{p_{i,j} \in \mathcal{P}_{-i}} -\langle R_{i,j}, p_{i,j}\rangle  -D_i(p_{i,j},\pi_{-i})\right),\\
    \end{split}
\end{align}
where $R_{i,j}$ corresponds to the $jth$ row of $R_i$.

We note that we differentiate the penalty functions $D_i$ to allow for the different agents to have different risk preferences. In this form, one can see that in a risk-adjusted game, the players imagine that intermediate adversaries seek to maximize their cost within a small ball in some metric of the opponents realized strategies. Thus, agents in risk-averse games introduce a certain amount of worst-case thinking into their strategies.  

In the aggregate case, from the point of view of player $i$ the game has simplified from a setting with $n-1$ opponents to a game with $1$ opponent who is penalized from taking actions that deviate too far from the \emph{product distribution} of the original $n-1$ opponents' strategies. In the action-dependent case, the game from player $i$'s point of view goes from a $n$-player game to a $A_i$-player game who are also similarly constrained.

Since the Nash equilibrium of the risk-adjusted game can be qualitatively different from that of the original game, we define a risk-averse Nash equilibrium (RNE) as follows. For ease of exposition we do not differentiate between aggregate and action-dependent RNE since similar results hold for either formulation.

\begin{definition}[Risk-adjusted Nash equilibrium] \label{def:RNE}
    A  risk-averse Nash equilibrium of is joint strategy $\pi^*\in \mathcal{P}$ such that no player has any incentive to unilaterally deviate in the risk adjusted game---i.e., for all $i=1,...,n$
    \begin{align*}
    f_i(\pi_i^*,\pi_{-i}^*)\le f_i(\pi_i ,\pi_{-i}^*) \ \ \forall \ \ \pi_i \in \Delta_{A_i}.
\end{align*}
\end{definition}
Note that since players would like to \emph{minimize} risk, the direction of the inequality has changed. The convexity and continuity of the penalty function guarantees that the risk-adjusted games admit at least one RNE. 

\begin{theorem}\label{thm:rne-existence}
    There always exists a RNE for all aggregate and action-dependent risk-averse games presented in~\eqref{eq:riskmat2} and~\eqref{eq:action_dep2} respectively.
\end{theorem}
\begin{proof}
   To begin, we show that $f_i(\pi_i,\pi_{-i})$ is convex in $\pi_i$ for all $\pi_i \in \mathcal{P}_{-i}$ in~\eqref{eq:riskmat2}. This follows from the fact that $D_i(\cdot,\cdot)$ is assumed to be convex and continuous in its first argument. Invoking Danskin's theorem guarantees us that $f_i$ is convex in $\pi_i$ for all fixed values of $\pi_{-i}$. Since the probability simplex is compact and convex, the game satisfies the conditions of a convex game~\cite{Rosen} and thus a Nash equilibrium must exist. The action-dependent risk aversion regime follows from the linearity of $f_i$ in $\pi_i$ in~\eqref{eq:action_dep2} and a similar invocation of the existence of Nash in convex games~\cite{Rosen}.
\end{proof}

The fact that these Nash equilibria exist is a consequence of the fact that players are not risk averse to the randomness that they introduce into the game by playing a mixed strategy. Under the stronger assumption that they are risk-averse to all randomness, previous work~\cite{fiat2010players} shows that RNE may not exist. Thus, our specific formulation of risk is crucial in that it will always yield a RNE.

Despite the fact that risk preferences already help convexify player's objectives, the addition of risk can serve to weaken existing structures in the player's cost function. To illustrate this, we show that even in two-player zero-sum games where $R_1=R=-R_2^T$, the risk-adjusted game loses any zero-sum structure and may cease to even be strictly competitive in the sense that one player's gain is the other's loss. 
\begin{example}\label{example:risk-averse-non-strictly-competitive}
   Consider a 2-player zero-sum game where $R_1=R=-R_2^T$ where players have aggregate risk aversion in the entropic risk metric with different degrees of risk aversion $\tau_1$ and $\tau_2$. Under these conditions, the players loss functions take the following form:
   \vspace{-0.5mm}
   \begin{align*}
       f_1(\pi_1,\pi_{2})&= \sup_{p_1 \in \mathcal{A}_{2}} -\pi_1^T R p_1 -\frac{1}{\tau_1}KL(p_1,\pi_{2}) =\frac{1}{\tau_1} \log\Big( \sum_{1\leq j\leq A_2} \pi_2(j)\exp(-\tau_1 [R\pi_1]_j)\Big) \\
       f_2(\pi_1,\pi_{2})&= \sup_{p_2 \in \mathcal{A}_{1}} \pi_2^T R^T p_2 -\frac{1}{\tau_2}KL(p_2,\pi_{1}) =\frac{1}{\tau_2} \log\Big( \sum_{1\leq j \leq A_1} \pi_1(j)\exp(\tau_2 [R^T\pi_2]_j)\Big). 
   \end{align*}
   Even instantiating $R=\mathbb{I}_2$, $\tau_1=10$, and for any $\tau_2>0$, both  $f_1(\pi_1,\pi_2)>f_1(\pi_1',\pi_2)$ and $f_2(\pi_1,\pi_2)>f_2(\pi_1',\pi_2)$ holds for the regions $\pi_1,\pi_2\in\Delta_2$ satisfying $\pi_2(1) \in (0.1,0.5)$ and $0.75-\pi_1(1)> \pi'_1(1)> \pi_1(1)$, which implies that the game is not a strictly competitive game.
\end{example}


Note that the previous example introduces degrees of risk aversion into the game through the parameters $\tau_1>0$ and $\tau_2>0$. As $\tau$ increases, the game becomes less reliant on the regularization term, which makes the adversary more powerful and results in more conservative game playing by the player. As $\tau$ goes to zero we recover the risk neutral regime~\cite{ahmadi2012entropic}.

Finally, note that the additional convexity induced by the introduction of risk aversion  guarantee is not enough to ensure the computational tractability of the Nash equilibrium (see e.g.,~\cite{mcmahan2024roping}). This is not surprising given the fact that computing NE in general convex games can be intractable and that further structural conditions are required for computational tractability. 

\subsection{Bounded Rationality in Matrix Games}\label{sec:bounded-rationality}

Since risk aversion on its own is not sufficient to guarantee computational tractability of NE, we also introduce a second important feature of human decision-making into the game: bounded rationality. Interestingly, in experimental settings risk aversion---by itself--- is also not enough to model human decision-making in game theoretic scenarios~\cite{goeree2003risk}.

To incorporate bounded rationality into agents, we resort to the notion of a \emph{quantal response function}. 

\begin{definition}\label{def:quantal-function}
    A \emph{quantal response function} is a 
\textit{continuous} function $\sigma:\mathbb{R}^n\rightarrow \Delta_n$ such that: If $x_i<x_j$, $\sigma_i(x)>\sigma_j(x)$, where $x_k,\sigma_k(x)$ represent the $k$-th components of $x$ and $\sigma$, respectively. 
\end{definition}


Clearly, when players responses are constrained to be \emph{quantal} responses, they cannot be perfect maximizers since the $\arg\max$ function does not satisfy the first desiderata of a quantal response function. Common quantal response functions include the logit response function:
\begin{align}
    \sigma(x)=\frac{\exp(-\frac{1}{\epsilon}  x_i)}{\sum_{j=1}^n \exp(-\frac{1}{\epsilon} x_j)},
\end{align}
where the sign is to account for the fact that agents may be \emph{minimizing} their loss function. 

When all players achieve a Nash equilibrium in the space of their quantal responses---the class of strategies that can be represented by a fixed class of quantal response functions (cf.~Definition~\ref{def:quantal-function})---the resulting equilibrium is known as a Quantal Response Equilibrium (QRE). To incorporate quantal responses into our risk-adjusted game, we introduce regularization to the player's losses. This can be shown to be equivalent to constraining the player's responses to quantal responses (see, e.g., \cite[Proposition 7]{Risk_overview}, or \cite{quantal1,quantal2}). 
We write a regularized risk-adjusted game as:
\begin{align}
\begin{split}\label{eq:RQEgame}
    f^{\epsilon_i}_i(\pi_i,\pi_{-i})=f_i(\pi_i,\pi_{-i}) +\epsilon_i\nu_i(\pi_i) \\
    \end{split}
\end{align}
where $\nu_i$ is strictly convex over the simplex and controls the class of quantal responses available to player $i$ and $\epsilon_i$ controls the agent's degree of bounded rationality. 

\begin{example}\label{example:quantal-response-function}
   If players are constrained to \emph{logit} response functions, one can reflect this by incorporating a negative entropy regularizer $\nu(\pi)=\sum_{i}p_i\log(p_i)$. Another class of quantal response functions would be generated by making use of e.g., a log-barrier regularizer $\nu(\pi)=-\sum_{i}\log(p_i)$. Both of these regularizers give rise to quantal response functions that satisfy Definition~\ref{def:quantal-function}.
\end{example}

This game now incorporates two key properties of human decision-making: risk aversion and bounded rationality on the part of the agents. The natural outcome of this game is what we term as a risk-adjusted quantal-response equilibrium (RQE). 
\begin{definition}
    A  risk-adjusted quantal response equilibrium (RQE) of a $n$-player general-sum finite-action game is a joint strategy $\pi^* \in \mathcal{P}$ such that for each player $i=1,...,n$:
       \begin{align*}
    f^{\epsilon_i}_i(\pi_i^*,\pi_{-i}^*)\le f^{\epsilon_i}_i(\pi_i ,\pi_{-i}^*) \ \ \forall \ \ \pi_i \in \Delta_{A_i}
\end{align*}
\end{definition}

For any set of convex regularizers $\nu_i$, it is easy to observe the that game remains a convex game and thus RQE will exist.

As we will show in the following section, if player's risk preferences and regularizers (and thus their set of quantal response functions) satisfy a simple relationship, one can show that \emph{irrespective of the underlying payoffs} $R_1,...R_n$, the game admits $RQE$ that can be efficiently computed using arbitrary no-regret learning algorithms which include e.g., gradient-play and mirror descent.

\section{Conditions for Computational Tractability of RQE}\label{sec:computation_matrix}
In this section we give conditions on the bounded rationality and degree of risk aversion of the players such that computing cost of RQE in the games is tractable.  Further, an immediate consequence of our result is that players can find such equilibrium outcomes in a decentralized manner using no-regret learning algorithms. 

\begin{remark}
    For ease of exposition we focus on deriving the results for the aggregate risk-aversion formulation and defer the results for action-dependent risk-aversion to Appendix~\ref{app:action_dep}. The results for both formulations in terms of computational tractability are similar in that both allow for a class of RQE to be learned using no-regret learning. However, we remark that the analysis of the aggregate case is comparatively simpler and allows one to compute simple bounds on the set of RQE that are computationally tractable. The action-dependent case requires more work and slightly more stringent requirements on the relationship between bounded rationality and risk-aversion.
\end{remark}

To derive our results on the computational tractability of RQE under aggregate risk-aversion, we first introduce a related $2n$-player game in which we associate to each original player $i$ an adversary whose strategy we denote $p_i \in \mathcal{P}_{-i}$. Let $p=(p_1,...,p_n)\in \bar{\mathcal{P}}=\prod_{i=1}^n \mathcal{P}_{-i}$. In this new game, each original player's loss function takes the form:
\begin{align}J_i(\pi_i,\pi_{-i},p)= -\pi_i^T R_i p_i -D_i(p_i,\pi_{-i}) +\epsilon_i\nu_i(\pi_i) \label{eq:main_loss}\end{align}
For each player $p_i$ we associate them to a new loss function which we denote:
\begin{align}
    \bar J_i(\pi,p_i,p_{-i})&=\pi_i^T R_i p_i +D_i(p_i,\pi_{-i}) -\sum_{j\ne i}\xi_{i,j}\nu_j(\pi_j). \label{eq:adversarial-loss}
\end{align}
In this $2n$ player game, the $i$-th player's adversary, whose strategy is $p_i$ seeks to minimize their loss $\bar{J}_i$ which is strategically the same as maximizing $J_i$. We note that unlike the original risk-adjusted game, the introduction of the new players gives us additional structure that we can exploit. A key feature that we introduce are  the additional parameters $\xi_{i,j}>0$  which we will use to characterize the degree of risk aversion of each player $i$.

Given the definition of the $2n$-player game we show how no-regret learning algorithms can be used to compute a RQE. To do so, we show that CCEs of the $2n$-player game coincide with NEs of it and NEs of the $2n$-player game coincide with RQE of the original $n$-player matrix game. This is a phenomenon known as equilibrium collapse which is well known to happen in zero-sum games and certain generalizations of zero-sum games~\cite{eq_collapse1,eq_collapse_Markov}. To prove this, we first define coarse correlated equilibria.

\begin{definition}
    A coarse correlated equilibrium of the $2n$-player game is a probability measure $\sigma$ on $\mathcal{P}\times \bar{\mathcal{P}}$ such that for all $i=1,...,n$:
    \begin{align*}
        \mb{E}_{(\pi,p)\sim\sigma}[ J_i(\pi,p)] &\le \mb{E}_{(\pi_{-i},p)\sim\sigma}[ J_i(\pi'_i,\pi_{-i},p)] \ \ \ \forall \pi'_i \in \Delta_{A_i}\\
        \mb{E}_{(\pi,p)\sim\sigma}[ \bar J_i(\pi,p)] &\le \mb{E}_{(\pi,p_{-i})\sim\sigma}[\bar J_i(\pi,p_i',p_{-i})] \ \ \ \forall p_i' \in \mathcal{P}_{-i}.
    \end{align*}
\end{definition}
With this definition in hand we now present our results on computational tractability of RQE in 2-player matrix games. We also extend the results to $n$-player setting in \cref{appen:n-player}. 

To simplify our results we define $H_1(p_1,\pi_2)=D_1(p_1,\pi_2) -\xi_1 \nu_2(\pi_2)$ and $H_2(p_2,\pi_1)=D_2(p_2,\pi_1) -\xi_2 \nu_1(\pi_1)$. Let $\xi_1^*>0$ and $\xi_2^*>0$ be the smallest values of $\xi_1$ and $\xi_2$ such that $H_1(p_1,\pi_2)$ and $H_2(p_2,\pi_1)$ are concave in $\pi_2$ and $\pi_1$ for all $p_1$ and $p_2$ respectively. The parameters $\xi^*_1$ and $\xi^*_2$ capture the player's degrees of risk aversion as illustrated in Corollaries~\ref{cor:KL} and~\ref{cor:reverse-KL} in Appendix~\ref{app:more} respectively.

The following theorem gives conditions on $\xi^*_1, \xi^*_2$ and $\epsilon_1, \epsilon_2$, the bounded rationality parameters under which RQE's are computationally tractable using no-regret learning. The proof is postponed to \cref{proof:thm:2player}.

\begin{theorem}\label{thm:2player}
    Assume the penalty functions that give rise to the players' risk preferences $D_1(\cdot,\cdot)$ and $D_2(\cdot,\cdot)$ are jointly convex in both their arguments. If $\sigma$ is a CCE of the four player game with $\xi_{1,2}=\xi_1^*$ and $\xi_{2,1}=\xi_2^*$, and
    \[ \frac{\epsilon_1}{\xi^*_1}\ge \frac{\xi_2^*}{\epsilon_2}, \] 
    then $\hat \pi_1=\mathbb{E}_\sigma[\pi_1]$ and $\hat \pi_2=\mathbb{E}_\sigma[\pi_2]$ constitute a RQE of the risk-adjusted game.
\end{theorem}

The previous theorem gives a range of risk aversion and bounded rationality parameters under which an RQE can be computationally tractable using no-regret learning. Before presenting corollaries for certain instantiations of risk measures and quantal response functions in \cref{app:more}, we briefly remark that the result does not necessarily guarantee uniqueness, though by exploiting the connections between socially convex games and monotone games~\cite{Gemp2017OnlineMG} such a result would follow.  

To validate this result, we show that the class of RQE that is computationally tractable seems sufficiently rich to capture the aggregate strategy played by people in various matrix games studied in laboratory experiments in behavioral economics~\cite{selfridge1989adaptive,goeree2003risk,selten2008stationary} to showcase differing human behavior from Nash equilibrium strategies. 
In particular, we showcase in \cref{fig:kl-rkl-plot} that for some choice of parameters ($\epsilon_j$ and $\tau_j$'s) satisfying conditions in \cref{thm:2player}, we recover risk-averse and rational human behaviors in various matching pennies games mentioned in \cref{tab:payoff_matrices-GHP,tab:payoff_matrices-SC}. 
We provide more details about the games and experiments in \cref{appen:exp:matching-pennies} for reproducing our results.

\section{Extension to Markov Games}\label{sec:Markov}

In this section, we extend our results to finite-horizon Markov games, or stochastic games \cite{shapley1953stochastic} which allow us to model dynamic games which are played out over Markov decision processes.

A Markov game can be seen as a sequential matrix game with dynamics. In this paper, we consider a general-sum finite-horizon risk-aversion Markov game involving $n$ players, represented as $ \mathcal{MG}=\{H, \cS, \{\cA_i\}_{i\in[n]}, \{ R_{i,h}, P_{i,h}\}_{i\in [n], h\in [H]}\}$. Here, $H$ is the time horizon length of the Markov games; $\cS = \{1,2,\cdots, S\}$ represents the state space of the joint environment with size $S$; we reuse several definitions in matrix game (see Section~\ref{sec:Matrix}), $\{\cA_i\}_{i\in[n]}$ is the action spaces with cardinality $|\cA_i| = A_i$, the joint action space is $\cA$, the joint action profile is $a \in \cA$, and $\cP$ is the product policy space; $R_{i,h}$ is the utility function of the $i$-th player at time step $h$ for any $(i,h)\in [n] \times [H]$, i.e., $R_{i,h}(s,\ba)$ represents the immediate reward (or utility) received by the $i$-th player given state-action pair $(s,\ba)\in\cS \times \cA$. For simplicity, we assume $R_{i,h}$ is deterministic; $P_{i,h}: \cS \times \cA \mapsto \Delta_{S}$ denote the transition probability so that $P_{i,h}(s' \mymid s, \ba)$ represents the probability transitioning from current state $s$ to the next state $s'$ conditioned on the action $\ba$. 

\paragraph{Markov policies and value functions.}  In this section, we focus on a class of Markov policy,
where at any state $s$ at any time step $h\in[H]$, the action selection rule depends only on the current state $s$, and is independent of previous trajectories. Specifically, for any $1\leq i\leq n$, the $i$-th agent executes actions according to a policy $ \pi_i = \{\pi_{i,h} :\cS   \mapsto \Delta_{A_i}  \}_{1\leq h\leq H}$, with $\pi_{i,h}(a \mymid s)$ the probability of selecting action $a$ in state $s$ at time step $h$.
We consider the joint policy of all agents as a product policy defined as $\pi= (\pi_1,\ldots, \pi_n): \cS \times [H] \mapsto  \cP $ and the corresponding policy space as $\Pi$, namely, the joint action profile $\bm{a}$ of all agents is chosen according to the distribution specified by $\pi_h(\cdot \mymid s) = (\pi_{1,h}, \pi_{2,h}\ldots, \pi_{n,h})(\cdot\mymid s) \in  \cP $ conditioned on state
$s$ at time step $h$.  For any given  $\pi$, we employ $\pi_{-i}: \cS \times [H] \mapsto \cP_{-i}$ to represent the policies of all agents excluding the $i$-th agent.

Then for any joint policy $\pi$, we can introduce the corresponding value function  $\widetilde{V}_{i,h}(\pi): \cS \mapsto \mathbb{R}$ (resp.~Q-function $\widetilde{Q}_{i,h}(\pi): \cS \times \ac \mapsto \mathbb{R}$) of the $i$-th agent to characterize the long-term cumulative reward by as follows: 
\begin{align}
\forall (s,a,h)\in \cA\times \cS \times [H]: \quad 	 \widetilde{V}_{i,h}(\pi ;s) &\coloneqq\mathbb{E}_{\pi}\left[\sum_{t=h}^{H} R_{i,t}\big(s_{t}, \bm{a}_{t}\big)\mid s_{h}=s\right], \nonumber \\
  \widetilde{Q}_{i,h}(\pi; s,\ba)&\coloneqq\mathbb{E}_{\pi}\left[\sum_{t=h}^{H}  R_{i,t}\big(s_{t}, \bm{a}_{t}\big)\mid s_{h}=s, \ba_h = \ba\right], \label{eq:value-function-defn}
\end{align}
where the expectation is taken over the Markovian trajectory $\{(s_t,\bm{a}_t)\}_{h\leq t\leq H}$ by executing the joint policy $\pi$ (i.e., $\bm{a}_t\sim \pi_t(\cdot \mymid s_t)$) over the Markov dynamics.

\subsection{Risk-Averse Markov games}

Given these definitions, now we are ready to introduce the risk-averse Markov games (RAMGs) --- considering both the uncertainty of agents' strategies and the environments (rewards and transitions). Towards this, we first introduce two important risk-measures for agents' strategies and environment respectively to form RAMGs together. Resorting to two sets of convex (convex w.r.t. the first argument with the second argument fixed) and lower-semi-continuous functions $\{D_{\mathsf{pol},i}(\cdot,\cdot)\}_{i\in[n]}$ and $\{D_{\mathsf{env},i}(\cdot,\cdot)\}_{i\in[n]}$, given any policy $\pi$, the two risk-measures are defined as: 
\begin{align}
  \forall i\in[n]: \quad f_{\mathsf{pol},i}^{\pi}(Q_i)&= \sup_{p_i \in \cP_i } -\pi_i^T Q_i p_i - D_{\mathsf{pol},i}(p_i,\pi_{-i}), \notag \\
  \forall i\in[n]: \quad f_{\mathsf{env},i}(R, P, V) & =  \inf_{\widetilde{P} \in \Delta_S } \left[R + \widetilde{P}V +D_{\mathsf{env},i}(\widetilde{P}, P)\right]. \label{eq:RQEgame-markov}
\end{align}
Here, $Q_i \in \mathbb{R}^{A_i \times \sum_{j\neq i} A_j }$ is the payoff matrix for the $i$-th player, $R\in\mathbb{R}^A$, $P\in\Delta_S$, $V\in\mathbb{R}^S$. 
Note that for simplicity, we choose the same $\{D_{\mathsf{pol},i}\}_{i\in[n]}$ and  $\{D_{\mathsf{env},i}\}_{i\in[n]}$ for all time steps $h\in[H]$, which can be extended to diverse ones for each time step.  


Armed with those two penalty functions $f_{\mathsf{pol},i}^{\pi}(\cdot)$ and $f_{\mathsf{env},i}(\cdot)$, in a RAMG, the $i$-th agent aim to minimize its own risk-averse long-term loss function $\{V_{i,h}(\pi)\}_{ h\in[H] }$, where $ V_{i,h}(\pi): \cS \mapsto \mathbb{R}$ for any joint strategy $\pi \in \Pi$  of the $i$-th agent is defined recursively as follows: for all $(h,s,\ba)\in [H]\times \cS \times \cA$,
\begin{align}\label{eq:RQEgame-markov-V}
\begin{split}
&V_{i,h}(\pi;s)=  f_{\mathsf{pol},i}^{\pi}\left(Q_{i,h}(\pi;s,:) \right) 
    \end{split}
\end{align}
where
\begin{align}
   &Q_{i,h}(\pi;s, \ba) = f_{\mathsf{env},i}\left( R_{i,h}(s, \ba), P_{h,s,\ba}, V_{i,h+1}(\pi) \right) \quad \text{ and } \quad P_{h,s,\ba} \defn P_h(\cdot \mymid s, \ba) \in \mathbb{R}^{1\times S}. \label{eq:RQEgame-markov-Q}
\end{align}
In words, $Q_{i,h}(\pi; s,:) \in \mathbb{R}^{A_i \times \sum_{j\neq i} A_j }$ is the payoff matrix of the $i$-th player at state $s$ and time step $h$, where the value at the $a_i$-th row and the $a_{-i}$-th column is specified by $Q_{i,h}(\pi; s, \ba)$ with $\ba = (a_i, a_{-i})$.
Therefore, the long-term loss function can be represented as a non-linear sum over the future immediate rewards as below:
\begin{align}
     &V_{i,h}(\pi; s) = f_{\mathsf{pol},i}^{\pi} \circ f_{\mathsf{env},i}\big( R_{i,h}(s,:), P_{h,s,:}, \big[ f_{\mathsf{pol},i}^{\pi} \circ f_{\mathsf{env},i}\big( R_{i,h+1} ,P_{h+1}, \big[f_{\mathsf{pol},i}^{\pi} \circ f_{\mathsf{env},i}\big( R_{i,h+2}, P_{h+2}, \cdots \big)   \big] \big)  \big] \big).
\end{align}

\paragraph{Solution concepts of RAMGs: risk-averse quantal response equilibrium (RQE).}
After introducing the RAMG problem formulation, we will focus on the solution concept that incorporate both risk-aversion and bounded rationality of human decision-making --- RQE. To specify, we first introduce the regularized value functions for any given $\{\epsilon_i\geq0\}_{i\in[n]}$:
\begin{align}
    \forall (i,h,s) \in [n]\times [H]\times \cS: \quad  V^{\epsilon_i}_{i,h}(\pi;s) &= \sup_{p_i \in \cP_i } -\pi_i^T Q_{i,h}(\pi; s,:) p_i - D_{\mathsf{pol},i}(p_i,\pi_{-i}) + \epsilon_i \nu_i(\pi_i) \notag\\
    & = V_{i,h}(\pi;s)  + \epsilon_i \nu_i(\pi_i). \label{eq:new-defined-value-epsilon}
\end{align}

In addition, for any policy $\pi$, we denote $\pi_{i,-h} \defn \{\pi_{i,t}\}_{t=1,2,3,\cdots, h-1,h+1,\cdots, H}$ as the $i$-th agent's strategy over all time steps except step $h$.

\begin{definition}[RQE of Markov games] \label{def:RQE-markov-games}
    A product policy $\pi=\pi_1\times \cdots \times \pi_n \in \Pi$ is 
said to be a risk-averse quantal response equilibrium (RQE) of RAMG if  we have
\begin{align}
   \forall (i,s,h)\in [n] \times \cS\times [H]: \quad V_{i,h}^{\epsilon_i}(\pi; s) \leq \min_{\pi_h': \cS \mapsto \Delta_{\cA_i}} V_{i,h}^{\epsilon_i}\big( (\pi_h', \pi_{i,-h})\times \pi_{-i} ;s \big),
    \label{eq:defn-Nash-E}
\end{align}
where $\nu_i(\cdot)$ is a strictly convex function over the simplex and controls the class of the $i$-th player's QRE with certain degree of bounded rationality $\epsilon_i \geq 0$. 
	
\end{definition}

\subsection{Algorithm and Theoretical Guarantees }

\paragraph{RAMGs with perfect information.}
Notice that the risk-aversion Markov games (RAMGs) are  sequential games with multiple rounds, it is natural to learn the game from the last step $h=H$ and then back forward step by step. Towards this, at each time step $h$, for all $ (i,h,s,\ba)\in  [n] \times [H] \times \cS\times \cA$, we define the (original) payoff matrices for all agents at any state $s$ at time step $h$ as 
\begin{align}
 \forall (i,s,h)\in [n]\times \cS\times [H]: \quad \widehat{Q}_{i,h}(s, \ba) = R_{i,h}(s, \ba)  + \inf_{\widetilde{P} \in \Delta_S }  \widetilde{P} \widehat{V}_{i,h+1} + D_{\mathsf{env},i}(\widetilde{P}, P_{h,s,\ba}), \label{eq:nvi-iteration}
\end{align}
which shall be used to define a risk-averse matrix game against other players' strategies later so that we can calculate a QRE for it with tractable computation cost. Here, we denote the computing routine of a RQE of a matrix game as $\mathsf{RQE}(\cdot)$. Armed with this, we summarize the entire algorithm to compute RQE for RAMGs in \cref{alg:nash-ramg}. The following theorem shows that computing RQE using \cref{alg:nash-ramg} is computationally tractable when certain conditions of risk-measure and bounded rationality are satisfied, which can be directly verified based on Theorem~\ref{thm:nplayer}. 


\begin{theorem}\label{thm:nplayer-markov}
    For any RAMG $\mathcal{MG}$, assume the penalty functions that give rise to the players' risk preferences $\{D_{\mathsf{pol},i}(\cdot,\cdot)\}$ are jointly convex in both of their arguments. If for any $(s,h)\in\cS \times [H]$, we can construct a $2n$-player game with the other $n$-additional players' payoff as \eqref{eq:adversarial-loss} with $\xi_{i,j}=\xi_{i,j}^*$ for all $(i,j) \in [n]\times [n]$. And for all $i=1,...,n$ we have
    $ \epsilon_i\ge \sum_{j\ne i} \xi_{j,i}^*$, then the output policy $\widehat{\pi}$ from Algorithm~\ref{alg:nash-ramg} is a RQE of $\mathcal{MG}$. 
\end{theorem}
    RAMGs can be seen as a sequential of matrix games conditioned on states. As a corollary of the results in matrix games (\cref{thm:nplayer}), the above theorem shows that the proposed solution concept --- RQE in RAMGs enables computational tractability with any well-defined risk-measures $\{D_{\mathsf{env},i} \}_{i\in[n]}$. 

\paragraph{RAMGs with imperfect environmental information.}
In the previous paragraph, we focus on the cases when the dynamics of the environment of RAMG --- reward and transition kernel are perfectly known. While in practice, especially in MARL tasks, the environments can be complicated and unknown that need to be learned by interacting with (sampling from) the environments. In this subsection, we focus on such scenarios and provide the theoretical guarantees of Algorithm~\ref{alg:nash-ramg}.

\begin{algorithm}[t]
\begin{algorithmic}[1]
	\STATE \textbf{Input:} reward function $\{R_{i,h}\}_{i\in[n] \times h\in[H]}$, transition kernel $\{P_{h}\}_{h\in[H]}$. 
	\STATE  \textbf{Initialization:} $\widehat{Q}_{i,h}(s,a)= 0$, $\widehat{V}_{i,h}(s)=0$ for all $(s,\ba,h) \in \cS\times \cA  \times [H+1]$. \\
   \FOR{$h = H, H-1, \cdots,1$}
		
		\FOR{$i= 1,2,\cdots,n$ and $s\in \cS, \ba\in \cA $}
			
		\STATE	Set $\widehat{Q}_{i,h}(s, \ba)$ according to \eqref{eq:nvi-iteration}. \label{eq:robust-q-estimate}
			\ENDFOR

		\FOR{$s\in \cS$}
		\STATE Get $\pi_{h}(s)=\{\pi_{i,h}(s)\}_{1 \leq i \leq n} \leftarrow  \mathsf{RQE} \Big( \big\{  f_{\mathsf{pol},i}^{\pi,\epsilon_i}(\widehat{Q}_{i,h}(s,:)  \big\}_{1 \leq i \leq n} \Big)$. \label{eq:nash-subroutine}
		 
		\STATE	Set $\widehat{V}_{i,h}(s) = f_{\mathsf{pol},i}^{\pi_h,\epsilon_i}\big( \widehat{Q}_{i,h}(s,:) \big) $.
		\ENDFOR
\ENDFOR
	
	\STATE \textbf{output:} $\widehat{\pi} = \{\pi_h\}_{1\leq h \leq H}$. 
 
\end{algorithmic}
\caption{ Computation method of RQE for risk-averse Markov games (RAMGs).}
 \label{alg:nash-ramg}
\end{algorithm}

We focus on a fundamental sampling mechanism --- assuming access to a generative model or a simulator \cite{kearns1998finite,agarwal2020model}, which allows us to collect $N$ independent samples for each state-action pair generated based on the true environment $\{R_{i,h},P_h\}_{i\in[n], h\in[H]}$:
\begin{align}
    \forall (s,\ba,j)\in \cS\times\cA \times [n]:\quad  s_{i, s, \ba,h} \overset{i.i.d}{\sim} P_h(\cdot \mymid s,\ba),\quad r^i_{j,s,\ba,h} = R_{j,h}(s,\ba), \quad i \in [N].
\end{align}
The total sample size is, therefore, $N_{\mathsf{all}} 
 \defn  NS\prod_{i\in[n]} A_i$.

In such imperfect information cases, we propose a model-based approach, which first constructs an empirical reward function and nominal transition kernel based on the collected samples and then applies Algorithm~\ref{alg:nash-ramg} to learn a QNE. First, the empirical reward function and transition kernel $\widehat{P} \in \mathbb{R}^{S\prod_{i=1}^n A_i\times S}$ can be constructed on the basis of the empirical frequency of state transitions, 
\begin{align}
   \forall (s, \ba, j, h)\in \cS\times \cA  \times [n] \times [H]: \quad   \widehat{R}_{j,h}(s,\ba) = r^1_{j,s,\ba,h} \quad \text{and} \quad \widehat{P}_{h}(s'\mymid s,\ba) \defn  \frac{1}{N} \sum\limits_{i=1}^N \mathds{1} \big\{ s_{i,s,\ba} = s' \big\} .\label{eq:empirical-P-infinite}
\end{align}

Then with such empirical reward and transition, we can apply the oracle  summarized in Algorithm~\ref{alg:nash-ramg}.
The following theorem provides the finite-sample guarantees of this model-based approach as below:
\begin{theorem}\label{thm:imperfect-mg}
    For any RAMG $\mathcal{MG}$, we consider penalty functions $\{ D_{\mathsf{env},i}(\cdot,\cdot)\}_{i\in[n]}$ are $L$-Lipschitz w.r.t the $\ell_1$ norm of the second argument with any fixed first argument. Applying Algorithm~\ref{alg:nash-ramg} with the estimated reward $\{\widehat{R}_{j,h}\}$ and transition kernels $\{\widehat{P}_{h}\}_{h\in[H]}$ as input, the output solution $\widehat{\pi}$ is an $\delta$-RQE of $\mathcal{MG}$. Namely, we have
    \begin{align}
    \max_{(i,s,h)\in [n] \times \cS\times [H] } \Big\{V_{i,h}^{\epsilon_i}(\widehat{\pi}; s) - \min_{\pi_h': \cS \mapsto \Delta_{\cA_i}} V_{i,h}^{\epsilon_i}((\pi_h', \widehat{\pi}_{i,-h})\times \widehat{\pi}_{-i} ; s) \Big\} \leq \delta
    \end{align}
    as long as the total number of samples satisfying $$
    N_{\mathsf{all}} = NS\prod_{i\in[n]} A_i \geq 8S\prod_{i\in[n]} A_i HL \sqrt{ \frac{S}{N}\log\big(2SH\prod_{i\in[n]} A_i/\delta\big)}.$$
\end{theorem}
The proof is postponed to Appendix~\ref{proof:thm:imperfect-mg}.
We remark that $L$ can be some constant for various penalty function, such as $L=1$ when $\{D_{\mathsf{env},i}\}_{i\in[n]}$ are defined as any $\ell_p$ norm including total variation (TV). 
We also remark that similar sample complexity guarantees can be shown when $\{D_{\mathsf{env},i}\}_{i\in[n]}$ are defined as $\phi$-divergence and allude to analyses ideas in \cite{panaganti22a,xu-panaganti-2023samplecomplexity,shi2022distributionally,shi2023curious}.
Note that computing exact RQE can be computationally expensive or may not be necessary in practice. This result provides the first finite-sample guarantees of computing RQE in RAMGs. We show that even when the environment is not perfectly known but needs to be learned by samples, we can achieve approximate-RQEs of RAMGs using sufficient finite samples of size $N_{\mathsf{all}}$.

\subsection{Experiments and Evaluation} \label{appen:exp:grid-world}

\begin{figure}[ht]
	\centering
	\includegraphics[width=1.0\linewidth]{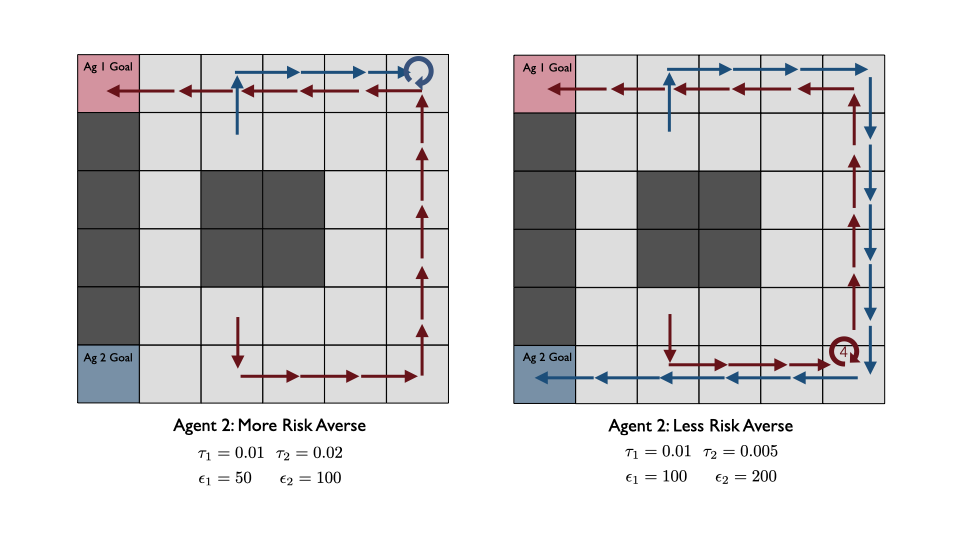}
    \vspace{-0.7cm}
	\caption{ \textit{Cliff Walk Description and Results}: The \textit{Cliff Walk} grid-world is depicted here  with color codes of grids: black is the cliff, blue and pink are agents/players 1 and 2's goals respectively, grey are grids for traversing.
    In both figures, red (blue) arrow depicts the actions taken by the learned policies for agent 1 (agent 2) reaching the states of the arrowhead. The agent 2 policy in the left figure showcases more risk-aversion (avoid each others' path to reduce the cliff risk) and less bounded rationality (not goal reaching).  The agent 2 policy in the right figure showcases less risk-aversion (more chances of cliff risk by sharing the path) and more bounded rationality (goal reaching). We also mention values of the risk-averse and bounded rational parameters $\epsilon_1,\epsilon_2$ and $\tau_1,\tau_2$ used in our experiments satisfying our theoretical conditions.
 }
	\label{fig:cliff}
\end{figure}

We consider a grid-world problem to test our algorithm (Algorithm~\ref{alg:nash-ramg}) for showcasing risk-averse and bounded rational behaviors. We note that we use our matrix game method (details described in \cref{appen:exp:matching-pennies}) for the $\mathsf{RQE}$ solver at Line 8 in Algorithm~\ref{alg:nash-ramg}. 

\paragraph{Cliff Walk Environment description:} A grid of size $6\times 6$ consists of some tiles that lead to the agents falling into the cliff.  
The environment setup is as follows: 
The cliff is the black grid with rewards $-2$. 
Agents/players are rewarded $0$ for taking each step and $1$ for when they reach the goal.
The agents move in the direction of action with probability $p_{\mathrm{d}}=0.9$ and takes uniform step in the unintended direction. To emphasize the risky choices of agents traversing along the cliff or traversing in each others' paths, we reduce $p_{\mathrm{d}}$ to $0.5$ when the agents are both within a grid apart. The episode horizon $H$ is $200$.

\paragraph{Algorithm details:}
We use the penalty function KL with the  log-barrier regularizer,
as mentioned in \cref{tab:risk-measure,example:quantal-response-function}, for our experiments. We also use KL divergence for the environmental uncertainty metric $D_{\mathsf{env},i}$ with the same bounded rational parameter $\tau_i^{-1}$ used by $\mathsf{RQE}$.
We use the full information of this grid-world to evaluate the Q-values described in Algorithm~\ref{alg:nash-ramg}. Additionally, as we use the matrix games solver for $\mathsf{RQE}$, we have similar statistical deviations for different training runs.
In \cref{appen:cliff-l1}, we also showcase results with $\ell_1$ metric for the environmental uncertainty metric $D_{\mathsf{env},i}$.

\paragraph{Results Discussion:}
We present two results in \cref{fig:cliff}. For both results, agent 1 starts at the 5th row and 3rd column of the cliff-walking grid-world, and agent 2 starts at the 2nd row and 3rd column grid position.
In the left figure of \cref{fig:cliff}, 
agent 1 is risk-averse to the environmental uncertainty induced by the cliff grids by traversing the longer route shown by the red arrow.
Interestingly, complementing our risk-adjusted notion, agent 2 is risk-averse w.r.t both the environmental uncertainty and evading the risk of coming into agent 1's path (which leads to more chances of falling into the cliff). It makes a less rational and more risk-averse choice by reaching the end state shown by the blue arrow.
In the case of the right figure of \cref{fig:cliff}. They both make a more rational choice of reaching goal states but take a riskier path, increasing the chances of falling into the cliff. Interestingly, agent 1 waits at the corner grid (at the bottom right) until agent 2 passes it, since being at this grid is a more risk-averse agnostic choice against other players.
These empirical evaluations support our theoretical result (\cref{thm:nplayer-markov}) w.r.t the relative choice of risk-averse and bounded rationality parameters ($\epsilon_j$ and $\tau_j$'s) for providing different level of risk-averse solutions.



\section{Conclusion}

By incorporating risk aversion and bounded rationality into agents' decision-making processes, we introduced a computationally tractable class of risk-averse quantal response equilibria (RQE) for matrix games and finite-horizon Markov games. This approach not only makes equilibrium computation tractable but also aligns well with observed human behavior in behavioral economics. 
We provided rigorous theoretical results for tractable computation of equilibrium concepts and also provided sample complexity results in the generative modeling setting for multi-agent reinforcement learning.
We also validated our findings on small-scale multi-agent game benchmarks: \textit{matching pennies games} and \textit{cliff walking environment}.

\section*{Acknowledgment}

The work of LS is supported in part by the Resnick Institute and Computing, Data, and Society Postdoctoral Fellowship at California Institute of Technology. 
KP acknowledges support from the `PIMCO Postdoctoral Fellow in Data Science' fellowship at the California Institute of Technology. EM acknowledges support from NSF Award 2240110.


\bibliographystyle{plain}
\bibliography{references}


\newpage
\input{appendix}

\end{document}

%% file: macro.tex
\newcommand{\defn}{\coloneqq}

\newcommand{\ac}{\mathcal{A}}

\newcommand{\ba}{\bm{a}}

\newcommand{\cA}{\mathcal{A}}

\newcommand{\cP}{\mathcal{P}}

\newcommand{\cS}{{\mathcal{S}}}

\newcommand{\mymid}{\,|\,} 

\usepackage{scalerel,stackengine}
\stackMath
\newcommand\reallywidehat[1]{%
\savestack{\tmpbox}{\stretchto{%
  \scaleto{%
    \scalerel*[\widthof{\ensuremath{#1}}]{\kern-.6pt\bigwedge\kern-.6pt}%
    {\rule[-\textheight/2]{1ex}{\textheight}}
  }{\textheight}%
}{0.5ex}}%
\stackon[1pt]{#1}{\tmpbox}%
}
\newcommand\reallywidecheck[1]{%
\savestack{\tmpbox}{\stretchto{%
  \scaleto{
    \scalerel*[\widthof{\ensuremath{#1}}]{\kern-.6pt\bigwedge\kern-.6pt}%
    {\rule[-\textheight/2]{1ex}{\textheight}}
  }{\textheight}%
}{0.5ex}}%
\stackon[1pt]{#1}{\scalebox{-1}{\tmpbox}}%
}

%% file: appendix.tex

\appendix

\section{Further results on computational tractability of RQE in aggregate risk-averse matrix games}\label{app:more}

In this section we provide further results on the aggregate risk aversion formulation of the risk-averse problem. This section provides a proof of the main result of Theorem~\ref{thm:2player} presented in the main body of the paper as well as the generalization of the results on computationaly tractability to $n$-players. A key step in the proof is relating the Nash equilibrium of the $2n$-player convex game presented where players have utilites of the form \eqref{eq:adversarial-loss} and \eqref{eq:main_loss} and RQE of the desired game~\eqref{eq:riskmat}. We then use this result to prove Theorem~\ref{thm:2player}.

\subsection{Relating Nash of the $2n$-player game and RQE}\label{app:nash}
We first note that this $2n$ player game is a convex game played over compact convex action sets and thus a Nash equilibrium must exist~\cite{Rosen}. To relate outcomes in this new game to that of the original game, we show that the strategies played by the original players in Nash equilibria of the $2n$-player game coincide with the RQE. 

\begin{proposition}\label{prop:equiv}
    If $(\pi^*,p^*)$ is a Nash equilibrium of the $2n$-player game, then $\pi^*$ is a RQE of Game~\eqref{eq:RQEgame}. Furthermore, if $\pi^*$ is a RQE of Game~\eqref{eq:RQEgame}, then  $(\pi^*,p^*)$ is a Nash equilibrium of the $2n$-player game, where:
    \begin{align}
        p_i^*=\arg\max_{p_i \in \mathcal{P}_{-i}} -{\pi^*_i}^T R_i p_i -D_i(p_i,\pi^*_{-i}).
    \end{align}
\end{proposition}

\begin{proof}
    
To prove this result, we rely on the definitions of Nash and RQE. We begin by proving the first claim of the proposition. Recall that a Nash equilibrium is a joint strategy $(\pi^*,p^*) \in \mathcal{P}\times \bar{ \mathcal{P}} $ such that, for all $i=1,...,n$:
\begin{align}
    J_i(\pi^*_i,\pi_{-i}^*,p^*)&\le J_i(\pi'_i,\pi_{-i}^*,p^*) \ \ \forall \ \ \pi_i' \in \Delta_{A_i}\label{eq:Ji}\\
    \bar J_i(\pi^*,p_i^*,p_{-i}^*)&\le \bar J_i(\pi^*,p_i',p_{-i}^*) \ \ \forall \ \ p_i' \in \mathcal{P}_{-i} \nonumber
\end{align} 
Noting that each $J_i(\pi^*_i,\pi^*_{-i}, p^*_i, p^*_{-i})$ does not depend on $p^*_{-i}$, to show the forward direction, we start by taking the supremum of the right hand side of ~\eqref{eq:Ji} over $p_i\in \mathcal{P}_{-i}$. Thus, we find that, for all $i$:
\begin{align*}
    J_i(\pi^*_i,\pi_{-i}^*,p^*)&\le \sup_{p_i \in \mathcal{P}_{-i}} J_i(\pi'_i,\pi_{-i}^*, p_i, p^*_{-i})\ \ \forall \ \ \pi_i' \in \Delta_{A_i}\\
    &=f_i^{\epsilon_i}(\pi'_i,\pi_{-i}^*) \ \ \forall \ \ \pi_i' \in \Delta_{A_i},
\end{align*}
where in the second line we used the fact that for any $ \pi,p_{-i}$, by definition, \[\sup_{p_i \in \mathcal{P}_{-i}} J_i(\pi,p_i,p_{-i})=f_i^{\epsilon_i}(\pi).\]

It remains to show that for any $i=1,...,n$, $J_i(\pi^*_i,\pi^*_{-i},p^*)=f_i^{\epsilon_i}(\pi^*_i,\pi_{-i}^*)$. This follows from the fact that the simplex is compact. Indeed for any fixed $\pi$, the function $J_i(\pi,p_i,p_{-i})$ is concave in $p_i$. Thus, the supremum is attained at $p_i^*$ since \[p_i^*=\arg\min_{p_i \in \mathcal{P}_{-i}}  \bar J_i(\pi^*,p_i,p_{-i}^*)=\arg\max_{p_i \in \mathcal{P}_{-i}} J_i(\pi^*,p_i,p_{-i}^*).\]

Since the same argument holds for all $i$, we have shown that if $(\pi^*,p^*)$ is a Nash equilibrium of the $2n$-player game, then:
\begin{align*}
    f^{\epsilon_i}_i(\pi_i^*,\pi_{-i}^*)&\le f^{\epsilon_i}_i(\pi_i ,\pi_{-i}^*) \ \ \forall \ \ \pi_i\in \Delta_{A_i}\\
\end{align*}
which is the definition of a RQE. 

To prove the second claim, suppose that $\pi^*$ is a RQE. By definition, we have that if \[p_i^*=\arg\max_{p_i \in \mathcal{P}_{-i}} -{\pi^*_i}^T R_i p_i -D_i(p_i,\pi^*_{-i}).\]
then $p_i^*$ by construction satisfies the condition of a Nash equilibrium of the $2n$-player game:
\[\bar J_i(\pi^*,p_i^*,p_{-i}^*)\le \bar J_i(\pi^*,p_i',p_{-i}^*) \ \ \forall \ \ p_i' \in \mathcal{P}_{-i}.\]
It remains to show that $\pi^*$ also satisfies the necessary conditions on $J_i$. To see that this must hold, by definition, we must have that $(\pi_i^*,p_i^*)$ satisfies:
\[J_i(\pi_i^*,p_i^*,\pi_{-i}^*,p_{-i}^*)=f_i^{\epsilon_i}(\pi^*_i,\pi_{-i}^*)=\min_{\pi_i\in \Delta_{A_i}}\max_{p_i \in \mathcal{P}_{-i}} J(\pi_i,p_i,\pi_{-i}^*,p_{-i}^*)\]
Further manipulations allow us to show that:
\begin{align*}
     J_i(\pi^*_i,\pi_{-i}^*,p^*)&= f_i^\epsilon(\pi^*_i,\pi_{-i}^*)\\
     &=\min_{\pi_i\in \Delta_{A_i}}\max_{p_i \in \mathcal{P}_{-i}} J(\pi_1,p_i,\pi_{-i}^*,p_{-i}^*)\\
     &=\min_{\pi_i\in \Delta_{A_i}}J(\pi_1,p_i^*,\pi_{-i}^*,p_{-i}^*)\\
     &\le J(\pi'_1,p_i^*,\pi_{-i}^*,p_{-i}^*) \ \ \forall \ \ \pi'_i\in \Delta_{A_i}
\end{align*}
Since this holds true for all $i$, this completes the proof.  
\end{proof}

\subsection{Proof of Theorem~\ref{thm:2player} and immediate corollaries}\label{proof:thm:2player}

Given the results showing that finding a Nash equilibrium of the $2n$-player convex game allows us to compute RQE, we present the proof of Theorem~\ref{thm:2player}. We also provide two corollaries that specialize the results to different risk-metrics and quantal response functions.

\begin{theorem}[Restatement of Theorem~\ref{thm:2player}]
 Assume the penalty functions that give rise to the players' risk preferences $D_1(\cdot,\cdot)$ and $D_2(\cdot,\cdot)$ are jointly convex in both their arguments. If $\sigma$ is a CCE of the four player game with $\xi_1=\xi_1^*$ and $\xi_2=\xi_2^*$, and
    \[ \frac{\epsilon_1}{\xi^*_1}\ge \frac{\xi_2^*}{\epsilon_2}, \] 
    then $\hat \pi_1=\mathbb{E}_\sigma[\pi_1]$ and $\hat \pi_2=\mathbb{E}_\sigma[\pi_2]$ constitute a RQE of the risk-adjusted game.
\end{theorem}

The proof follows by showing that CCE of the $2n$-player convex game coincide with Nash equilibria of the $2n$-player convex game, and then invoking results from the previous section. 

\begin{proof}
To prove this, we show that $(\hat \pi_1,\hat p_1,\hat \pi_2,\hat p_2)$ is a Nash equilibrium of the four player game and then invoke Proposition~\ref{prop:equiv} (where $\hat \pi_1=\mathbb{E}_\sigma[\pi_1]$, $\hat p_2=\mathbb{E}_\sigma[p_2]$,$\hat \pi_2=\mathbb{E}_\sigma[\pi_2]$, and $\hat p_2=\mathbb{E}_\sigma[p_2]$).

To begin, we focus on $J_1$ and $\bar J_1$. By symmetry, the same arguments hold for $J_2$ and $\bar J_2$. We first note that $J_1$ is (strictly) convex in $\pi_1$ for all fixed $p_1,\pi_2,p_2$ and jointly concave in $p_1,\pi_2,p_2$ for all fixed $\pi_1$. Now, starting with the definition of CCE for $\sigma$ and via Jensen's inequality, we have that:


\begin{align}
  \mb{E}_\sigma[ J_1(\pi_1,p_1,\pi_2,p_2)] &\leq  \mb{E}_{(p_1,\pi_2,p_2)\sim\sigma}[ J_1(\pi_1',p_1,\pi_2,p_2)] \notag \\
  & = \mb{E}_{p_1 \sim \hat{p}_1}\left[ \mb{E}_{ \pi_2 \sim \sigma \mymid p_1} \left[   \mb{E}_{p_2 \sim \sigma \mymid p_1, \pi_2} \left[J_1(\pi_1',p_1,\pi_2,p_2)  \right] \right] \right] \notag \\
  & \leq \mb{E}_{p_1 \sim \hat{p}_1}\left[ \mb{E}_{\pi_2 \sim \sigma \mymid p_1} \left[ J_1(\pi_1',p_1,\pi_2, \mb{E}_{\sigma \mymid p_1, \pi_2}[p_2])  \right] \right] \notag \\
  & \leq \mb{E}_{p_1 \sim \hat{p}_1}\left[  J_1(\pi_1',p_1, \mb{E}_{\sigma \mymid p_1}[\pi_2], \mb{E}_{\sigma \mymid p_1 }[p_2])  \right] \notag \\
  & \leq J_1(\pi_1', \hat{p}_1, \hat{\pi}_2, \hat{p}_2)\quad\forall \pi'_1 \in \Delta_{A_1}.\label{eq:cce-jensen-fact}
\end{align}

Similarly, we note that for $\xi_1=\xi_1^*$, $\bar J_1$ is jointly concave in $\pi_1,\pi_2,p_2$, such that:

\[ \mb{E}_\sigma[ \bar J_1(\pi_1,p_1,\pi_2,p_2)] \le  \bar J_1(\hat\pi_1, p_1',\hat \pi_2,\hat p_2)\ \ \ \forall p'_1 \in \Delta_{A_2}. \]

Letting $\hat z=( \hat \pi_1,\hat p_1,\hat \pi_2,\hat p_2)$ and $z=( \pi_1,p_1,\pi_2,p_2)$ to simplify notation, we can now take a weighted sum of the four utility functions with $\lambda \in (0,1)$ to find that:
\begin{align*}
    &\frac{\lambda}{2} \left(J_1(\hat z)+\bar J_1( \hat z)\right)+\frac{1-\lambda}{2}\left( J_2(\hat z)+\bar J_2( \hat z)\right)\\
    &\qquad \ge \mb{E}_\sigma\left[ \frac{\lambda}{2} \left(J_1(z)+\bar J_1(  z)\right)+\frac{1-\lambda}{2}\left( J_2( z)+\bar J_2( z)\right)\right] \\
    &\qquad = \frac{1}{2}\left(\lambda\epsilon_1-(1-\lambda)\xi^*_2\right)\mb{E}_\sigma\left[ \nu_1(\pi_1) \right] +\frac{1}{2}\left((1-\lambda)\epsilon_2-\lambda\xi^*_1\right)\mb{E}_\sigma\left[ \nu_2(\pi_2) \right].
\end{align*}
Choosing $\lambda={\xi^*_2}/{(\epsilon_1 +\xi^*_2)}$, we can further simplify to find that:
\begin{align*}
    &\frac{\lambda}{2} \left(J_1(\hat z)+\bar J_1( \hat z)\right)+\frac{1-\lambda}{2}\left( J_2(\hat z)+\bar J_2( \hat z)\right)\\
    &\qquad \ge \frac{1}{2}\left(\frac{\epsilon_1\epsilon_2}{\epsilon_1+\xi^*_2}-\frac{\xi^*_1\xi^*_2}{\epsilon_1+\xi^*_2}\right)\mb{E}_\sigma\left[ \nu_2(\pi_2) \right]\\
    &\qquad \ge \frac{1}{2}\left(\frac{\epsilon_1\epsilon_2}{\epsilon_1+\xi^*_2}-\frac{\xi^*_1\xi^*_2}{\epsilon_1+\xi^*_2}\right)\nu_2(\hat \pi_2)\\
    &\qquad =\frac{\lambda}{2} \left(J_1(\hat z)+\bar J_1( \hat z)\right)+\frac{1-\lambda}{2}\left( J_2(\hat z)+\bar J_2( \hat z)\right),
\end{align*}
where we used the fact that $\frac{\epsilon_1}{\xi^*_1}\ge \frac{\xi_2^*}{\epsilon_2} $ by assumption and invoke Jensen's inequality for $\nu_2$ at the second inequality.
Thus we have shown that:
\begin{align} 
&\frac{\lambda}{2} \left(J_1(\hat z)+\bar J_1( \hat z)\right)+\frac{1-\lambda}{2}\left( J_2(\hat z)+\bar J_2( \hat z)\right)\\
&\qquad = \frac{\lambda}{2} \left(\mb{E}_\sigma\left[J_1(z)\right]+\mb{E}_\sigma\left[\bar J_1(  z)\right]\right)+\frac{1-\lambda}{2}\left( \mb{E}_\sigma\left[J_2(z)\right]+\mb{E}_\sigma\left[\bar J_2(  z)\right]\right). \label{eq:bounded-seq-1} 
\end{align}
By \cref{eq:cce-jensen-fact}, we observe: \begin{align}\mb{E}_\sigma\left[J_1(z)\right]\leq J_1(\hat z), 
    \mb{E}_\sigma\left[\bar J_1(z)\right]\leq\bar J_1(\hat z), \mb{E}_\sigma\left[J_2(z)\right]\leq J_2(\hat z),
    \mb{E}_\sigma\left[\bar J_2(z)\right]\leq\bar J_2(\hat z). \label{eq:bounded-seq-2} \end{align}
We note a fact: $\lambda a+(1-\lambda)b=\lambda c+(1-\lambda)d$ and $a\le c, b\le d$ implies $a= c, b= d$ for any $a,b,c,d\in\mb R$. Using this fact for \cref{eq:bounded-seq-1,eq:bounded-seq-2}, we have:
\begin{align*}
    \mb{E}_\sigma\left[J_1(z)\right]&=J_1(\hat z)\le J_1( \pi'_1,\hat p_1,\hat \pi_2,\hat p_2) \ \ \ \forall \pi'_1 \in \Delta_{A_1}\\
    \mb{E}_\sigma\left[\bar J_1(z)\right]&=\bar J_1(\hat z) \le \bar J_1(\hat\pi_1, p_1',\hat \pi_2,\hat p_2)\ \ \ \forall p'_1 \in \Delta_{A_2}\\
    \mb{E}_\sigma\left[J_2(z)\right]&=J_2(\hat z)\le J_2( \hat \pi_1,\hat p_1,\pi_2',\hat p_2) \ \ \ \forall \pi'_2 \in \Delta_{A_2}\\
    \mb{E}_\sigma\left[\bar J_2(z)\right]&=\bar J_2(\hat z)\le \bar J_2(\hat\pi_1, \hat p_1,\hat \pi_2,p_2')\ \ \ \forall p'_2 \in \Delta_{A_1}
\end{align*}
Thus we have shown that  $\hat z=( \hat \pi_1,\hat p_1,\hat \pi_2,\hat p_2)$ is a Nash equilibrium for the 4-player game. By invoking Propostion~\ref{prop:equiv} we derive our result that $(\hat \pi_1,\hat \pi_2)$ must be a RQE for the original risk-adjusted game.
\end{proof}

%

 We now present two corollaries that specialize the results to specific risk metrics and quantal responses. In the first we look at the case where players make use of the entropic risk and log-barrier reguarlizers.

\begin{corollary}\label{cor:KL}
    Suppose the players are risk-averse in the entropic risk metric with parameters $\tau_1$ and $\tau_2$ respectively, meaning that their risk-adjusted losses are given by:
       \begin{align*}
       f_1(\pi_1,\pi_{2})= \sup_{p_1 \in \mathcal{A}_{2}} -\pi_1^T R_1 p_1 -\frac{1}{\tau_1}KL(p_1,\pi_{2})  \qquad f_2(\pi_1,\pi_{2})= \sup_{p_2 \in \mathcal{A}_{1}} \pi_2^T R_2 p_2 -\frac{1}{\tau_2}KL(p_2,\pi_{1}) 
   \end{align*}
    If they respond in the space of quantal responses generated by the log-barrier regularizers with parameters $\epsilon_1$ and $\epsilon_2$ respectively, and if $\epsilon_1 \tau_1 \ge \frac{1}{\epsilon_2\tau_2}$ then, for any $R_1,R_2$ the players can compute a RQE by using no-regret learning.
\end{corollary}

In the second corollary we look at the case when players make use of the reverse KL penalty function and logit quantal responses.
\begin{corollary}\label{cor:reverse-KL}
    Suppose the players are risk-averse and make use of the reverse-KL as a penalty function to give rise to their risk metric with parameters $\tau_1$ and $\tau_2$ respectively. Their risk-adjusted losses are given by:
       \begin{align*}
       f_1(\pi_1,\pi_{2})= \sup_{p_1 \in \mathcal{A}_{2}} -\pi_1^T R_1 p_1 -\frac{1}{\tau_1}RKL(p_1,\pi_{2})  \ \ f_2(\pi_1,\pi_{2})= \sup_{p_2 \in \mathcal{A}_{1}} \pi_2^T R_2 p_2 -\frac{1}{\tau_2}RKL(p_2,\pi_{1}) 
   \end{align*}
    If they respond in the space of quantal responses generated by the negative entropy regularizor with parameters $\epsilon_1$ and $\epsilon_2$ respectively, and if $\epsilon_1 \tau_1 \ge \frac{1}{\epsilon_2\tau_2}$ then, for any $R_1,R_2$ the players can compute a RQE by using no-regret learning.
\end{corollary}

\begin{proof}
    The proof of these corollaries result comes from the fact that for $\xi^*=\frac{1}{\tau}$, the function $H(p,\pi)=\frac{1}{\tau}RKL(p,\pi)-\xi \nu(\pi)$ is concave in $\pi$ for all $p$ if $\xi\ge\frac{1}{\tau}$. Thus, choosing $\xi^*=\frac{1}{\tau}$ and invoking Theorem~\ref{thm:2player} completes the proof.
\end{proof}

Note that the proof of this corollary is the same as that of Corollary~\ref{cor:KL} and so we only provide one proof for both. 

\subsection{Computing RQE in n-player General-Sum Games} \label{appen:n-player}

We now extend our result to the computation of RQE in $n$-player games. This requires stronger assumptions on the players' risk preferences and bounded rationality parameters. Nevertheless, we once again show that a large class of RQE is computationally tractable in this class of games.

To do so we now define $H_i(p_i,\pi_{-i})=D_i(p_i,\pi_{-i}) -\sum_{j\ne i} \xi_{i,j} \nu_j(\pi_j)$. For all $i,j \in \{1,...,n\}$ let $\xi_{i,j}^*>0$ be the smallest values of $\xi_{i,j}$  such that $H_i(p_i,\pi_{-i})$ is concave in $\pi_{j}$. Again, the parameters $\xi^*_{i,j}$ capture the player's degrees of risk aversion. The following theorem gives a general condition under which an RQE is computable using no-regret learning. 

\begin{theorem}\label{thm:nplayer}
    Assume the penalty functions that give rise to the players' risk preferences $D_i(\cdot,\cdot)$ are jointly convex in both of their arguments. If $\sigma$ is a CCE of the $2n$-player game with $\xi_{i,j}=\xi_{i,j}^*$ for all $i,j \in \{1,...,n\}$, and for all $i=1,...,n$ we have
    $ \epsilon_i\ge \sum_{j\ne i} \xi_{j,i}^*$, then $\hat \pi=\mb{E}_\sigma [\pi]$ is a RQE of the risk-adjusted $n$-player game.
\end{theorem}

\begin{proof}


To prove this, we show that $(\hat \pi,\hat p)$ is a Nash equilibrium of the $2n$-player game and then invoke Proposition~\ref{prop:equiv} (where $\hat \pi=\mathbb{E}_\sigma[\pi]$, $\hat p=\mathbb{E}_\sigma[p]$).

We focus on $J_i$ and $\bar J_i$.  We first note that $J_i$ is (strictly) convex in $\pi_i$ for all fixed $p,\pi_{-i}$ and jointly concave in $p,\pi_{-i}$ for all fixed $\pi_1$ by assumption. Thus, via Jensen's inequality, we have that:

\[ \mb{E}_\sigma[ J_i(\pi,p)] \le  J_i( \pi'_i,\hat \pi_{-i},\hat p) \ \ \ \forall \pi'_i \in \Delta_{A_i} \]

Similarly, we note that for $\xi_{i,j}=\xi_{i,j}^*$, $\bar J_i$ is jointly concave in $\pi,p_{-i}$, such that:

\[ \mb{E}_\sigma[ \bar J_i(\pi,p)] \le  \bar J_i(\hat\pi, p_i',\hat p_{-i})\ \ \ \forall p'_i \in \mathcal{P}_{-i}. \]

We can now take a sum of the $2n$ utility functions to find that:
\begin{align*}
    \sum_{i=1}^n J_i(\hat \pi,\hat p)+\bar J_i(\hat \pi,\hat p) &\ge \sum_{i=1}^n \mb{E}_\sigma[J_i(\pi,p)]+ \mb{E}_\sigma[\bar J_i(\pi,p)]\\
    &\ge \sum_{i=1}^n\left(\epsilon_i-\sum_{j\ne i} \xi^*_{j,i}\right)\mathbb{E}_\sigma[\nu_i(\pi_i)]\\
    &\ge \sum_{i=1}^n\left(\epsilon_i-\sum_{j\ne i} \xi^*_{j,i}\right)\nu_i(\hat \pi_i)\\
    &=\sum_{i=1}^n J_i(\hat \pi,\hat p)+\bar J_i(\hat \pi,\hat p)
\end{align*}
where we used the assumed condition on $\epsilon_i$ to guarantee convexity of the functions of $\pi_i$ in the second line, allowing us to use Jensen's inequality to derive the third line. 

Thus we have shown that:
\begin{align*} 
 \sum_{i=1}^n J_i(\hat \pi,\hat p)+\bar J_i(\hat \pi,\hat p)=\sum_{i=1}^n \mb{E}_\sigma[J_i(\pi,p)]+ \mb{E}_\sigma[\bar J_i(\pi,p)]
\end{align*}
This implies that:
\begin{align*}
    J_i(\hat \pi,\hat p)&=\mb{E}_\sigma[J_i(\pi,p)]\le  J_i( \pi'_i,\hat \pi_{-i},\hat p) \ \ \ \forall \pi_i' \in \Delta_{A_i}\\
    \bar J_i(\hat \pi,\hat p)&=\mb{E}_\sigma[\bar J_i(\pi,p)]\le  \bar J_i(\hat\pi, p_i',\hat p_{-i})\ \ \ \forall p'_i \in \mathcal{P}_{-i}.
\end{align*}
Thus we have shown that  $(\hat \pi,\hat p)$ is a Nash equilibrium for the $2n$-player game. By invoking Proposition~\ref{prop:equiv} we can observe that  $\hat \pi$ must be a RQE for the risk-adjusted game.
\end{proof}

\section{Computational tractability of RQE in action-dependent risk-averse matrix games}\label{app:action_dep}

As in the case of aggregate risk-aversion, to prove our results we again introduce an auxiliary game that we relate to our risk-adjusted game of interest. In this case, the loss of the original players is given by:
\begin{align}J_i(\pi_i,\pi_{-i},p)= \sum_{j\in \mathcal{A}_i}\pi_i(j)\left(-\langle R_{i,j} ,p_{i,j}\rangle -D_i(p_{i,j},\pi_{-i}) \right)+\epsilon_i\nu_i(\pi_i). \label{eq:action_dep_loss}\end{align}
We now associate each player $i$ to its intermediate adversary $p_i$. For each player $p_{i}$  their loss function is given by:
\begin{align}
    \bar J_{i}(\pi,p_{i})&=\sum_{j \in \mc{A}_i}\pi_i(j)\left(\langle R_{i,j},p_{i,j}\rangle +D_i(p_{i,j},\pi_{-i})\right)-\sum_{k}\xi_{i,k}\nu_k(\pi_k),
\end{align}
where $p_i =\{ p_{i,j}\}_{j\in \mathcal{A}_i}$ where $p_{i,j}\in \mathcal{P}_{-i}$.
This is once again a convex game since each player's loss is convex in its own argument. Define $\xi_{i,k}^*\ge0$ as the minimum value of $\xi_{i,k}$ needed for $\bar J_{i}(\pi,p_i)$ to be concave in $\pi$ for all values of $p_{i}$ for all $i$. Note that due to the structure of $\bar J$ the values of $\xi_{i,k}^*$ only depend on properties of the risk metrics under consideration which are captured in $D_i$, and \emph{not} on the payoff structure $R_i$.

The following proposition relates the Nash equilibrium of the $2n$-player convex game to the RQE of the action-dependent risk averse game.
\begin{proposition}\label{prop:action_equiv}
    If $(\pi^*,p^*)$ is a Nash equilibrium of the $2n$-player game, then $\pi^*$ is a RQE of Game~\eqref{eq:RQEgame}. Furthermore, if $\pi^*$ is a RQE of the action-dependent risk averse game, then  $(\pi^*,p^*)$ is a Nash equilibrium of the $2n$-player game, where:
    \begin{align}
        p_{i,j}^*=\arg\min_{p_{i,j} \in \mathcal{P}_{-i}} \bar J_{i,j}(\pi^*,p_{i,j})
    \end{align}
\end{proposition}

The proof of this result follows by exactly the same arguments as the proof of Proposition~\ref{prop:equiv} and is therefore omitted for brevity. 

Given these results we now prove the equivalent of Theorem~\ref{thm:nplayer} for the action-dependent formulation. We note that by a more careful accounting of terms, a stronger guarantee is possible in the $2$-player regime similar to Theorem~\ref{thm:2player}.

\begin{theorem}
 Assume the penalty functions that give rise to the players' risk preferences $D_i(\cdot,\cdot)$ are jointly convex in both their arguments. If $\sigma$ is a CCE of the $n$-player game and for each $i$: \[\epsilon_i\ge \sum_{j} \xi_{i,j}^*\]
    then the marginal strategies $\hat \pi_i=\mathbb{E}_\sigma[\pi_i]$ constitute a RQE of the action-dependent risk-averse game.
\end{theorem}

 Note that one point of departure for this result from the previous results is that $\sigma$ is now the CCE of the original 2-player convex game defined on the objective functions $f_1^{\epsilon_1}$ and $f_2^{\epsilon_2}$.

\begin{proof}
To begin, let $\sigma$ be a CCE of the convex game played on $f_i^{\epsilon_i}$, where:
 \[f^{\epsilon_i}_i(\pi_i,\pi_{-i})= \sum_{j \in \mathcal{A}_i} \pi_i(j)\left(\sup_{p_{i,j} \in \mathcal{P}_{-i}} -\langle R_{i,j}, p_{i,j}\rangle  -D_i(p_{i,j},\pi_{-i})\right) +\epsilon_i\nu(\pi_i).\]\\
 By definition of a CCE, we must have that:
 \[ \mathbb{E}_\sigma[f^{\epsilon_i}_i(\pi_i,\pi_{-i})]\le \mathbb{E}_\sigma[f^{\epsilon_i}_i(\pi_i',\pi_{-i})] \ \ \forall \ \pi_i'\in \Delta_{A_i},\]
for all $i=1,...,n$. Given $\sigma$, we can define a new distribution $\sigma'$ as $(\pi,p^*(\pi))$ where 
\begin{align*}
        p_{i}^*(\pi)=\arg\min_{p_{i} \in \mathcal{P}_{-i}} \bar J_{i}(\pi,p_{i}),
    \end{align*}
with $\xi_{i,j}=\xi_{i,j}^*$. We now claim that, by construction, $\sigma'$ is a CCE of the $2n$-player game. To see this, we observe that
\begin{align*}
    \mathbb{E}_{\sigma'}[J_i(\pi_i,\pi_{-i},p)]=\mathbb{E}_\sigma[f^{\epsilon_i}_i(\pi_i,\pi_{-i})]&\le \mathbb{E}_\sigma[f^{\epsilon_i}_i(\pi_i',\pi_{-i})]=\mathbb{E}_{\sigma'}[J_i(\pi_i',\pi_{-i},p)] \ \ \forall \ \pi_i \in \Delta_{A_i},
\end{align*}
where the first equality follows by construction of $\sigma'$, the second from the definition of a CCE, and the third from the fact that the value of $p^*(\pi)$ only depends on $\pi_{-i}$ and not $\pi_i'$. 

Similarly, we can show the same result for all the players $p_{i}$. By simply applying Jensen's inequality, we can see that
\begin{align*}
    \mathbb{E}_{\sigma'}[\bar J_{i}(\pi,p_{i})]=\mathbb{E}_{\pi\sim \sigma}\left[\min_{p_{i,j}} \bar J_{i}(\pi,p_{i,j})\right]\le \mathbb{E}_{\sigma'}[\bar J_{i}(\pi,p'_{i})] \ \ \forall p'_{i} \in \mathcal{P}_{-i}.
\end{align*}
Thus we can observe that $\sigma'$ is a CCE of the $2n$-player game. To show that the marginals of $\pi$ in this CCE are Nash equilibria of the convex game (and thus RQE), we proceed as before. Using the fact that $D(p_{i,j},\pi_{-i})$ is jointly convex in each of its arguments, we can apply Jensen's inequality to find that:
  \[\mb{E}_{\sigma'}[J_i(\pi_i,\pi_{-i},p)]\le J_i(\pi_i',\hat \pi_{-i},\hat p) \ \ \ \forall \pi_i' \in \Delta_{A_i},\] 
  where $\hat p_{i,j}=\mb{E}_{\sigma'}[p_{i,j}]$ and $\hat \pi_{i}=\mb{E}_{\sigma'}[\pi_{i}]$. Similarly, by our choice of $\xi_{i,j}=\xi_{i,j}^*$ we can find that:
  \[\mathbb{E}_{\sigma'}[\bar J_{i}(\pi,p_{i})]\le \bar J_{i}(\hat \pi,p_{i}') \ \ \ \forall p_{i}' \in \mc{P}_{-i}.\] 
  Now, we can observe that:
  \begin{align*}
      \sum_{i}J_i(\hat \pi_i,\hat \pi_{-i},\hat p) +\bar J_{i}(\hat \pi,\hat p_{i}) &\ge \mathbb{E}_{\sigma'}\left[ \sum_{i}J_i( \pi_i,\pi_{-i},p) +\bar J_{i}( \pi,p_{i})\right]\\
      &=\sum_i \mathbb{E}_{\sigma'}\left[\left(\epsilon_i-\sum_{j}\xi_{i,j}^*\right)\nu_i(\pi_i)\right]\\
      &\ge \sum_i\left(\epsilon_i-\sum_{j}\xi_{i,j}^*\right)\nu_i(\hat\pi_i)\\
      &=\sum_{i}J_i(\hat \pi_i,\hat \pi_{-i},\hat p) +\bar J_{i}(\hat \pi,\hat p_{i}),
  \end{align*} where the third inequality follows by the Jensen's inequality.
  By the same rationale as in the proof of Theorem~\ref{thm:nplayer} we can now conclude that $\hat \pi$ is a Nash equilibrium-joint strategy profile-of the $2n$ player game and thus (due to Proposition~\ref{prop:action_equiv}) an RQE of the action-dependent risk averse game.
\end{proof}

Similar to the aggregate risk regime, one can instantiate the previous theorem with specific risk-metrics and quantal response functions to illustrate the class of action-dependent RQE that are guaranteed to be computationally tractable. Importantly, this class once again only depends on the levels of risk averse and quantal response but \emph{not} on the underlying structure of the game. However, note that the requirements are more stringent since the requirements on the $\xi^*$'s are stronger since they must ensure joint convexity of $\bar J_i$. Nevertheless further algebraic manipulations would allow one to recover analogues of Corollaries~\ref{cor:KL} and~\ref{cor:reverse-KL} for the action-dependent risk case as well. For brevity we leave these as exercises to the reader.










\section{Proof of Theorem~\ref{thm:imperfect-mg}}\label{proof:thm:imperfect-mg}

Armed with the estimated reward and transition kernel in \eqref{eq:empirical-P-infinite}, we can construct an empirical MG $ \widehat{\mathcal{MG}}=\{H, \cS, \{\cA_i\}_{i\in[n]}, \{ \widehat{R}_{i,h}, \widehat{P}_{i,h}\}_{i\in [n], h\in [H]}\}$. Analogously, for any joint policy $\pi$, we denote the corresponding risk-averse loss functions or the payoff matrices as $\{\widehat{V}_{i,h}^{\epsilon_i}(\pi)\}$ and $\{\widehat{Q}_{i,h}^{\epsilon_i}(\pi)\}$, respectively. 


To begin with, recall the goal is to show that
\begin{align}
   \max_{(i,s,h)\in [n] \times \cS\times [H] } \Big\{V_{i,h}^{\epsilon_i}(\widehat{\pi}; s) - \min_{\pi_h': \cS \mapsto \Delta_{\cA_i}} V_{i,h}^{\epsilon_i}((\pi_h', \widehat{\pi}_{i,-h})\times \widehat{\pi}_{-i} ; s) \Big\} \leq \delta.
\end{align}

For convenience, we denote
\begin{align}
  \forall  (i,h,s) \in [n] \times [H] \times \cS: \quad V_{i,h}^{\star, \epsilon_i}(\widehat{\pi}_{-i}; s)  = \min_{\pi_{i,h}':\cS  \mapsto \Delta_{A_i}} V_{i,h}^{\epsilon_i}((\pi_{i,h}', \widehat{\pi}_{i,-h}) \times \widehat{\pi}_{-i} ; s)
\end{align}
and define the best-response policy $\pi^\star_i \defn \{\pi^\star_{i,h}: \cS \mapsto \Delta_{A_i}\}_{h\in[H]}$ so that
\begin{align}
    \forall  (i,h,s) \in [n] \times [H] \times \cS: \quad \pi^\star_{i,h}(s) \defn \mathrm{argmin}_{\pi_i':\cS  \mapsto \Delta_{A_i}} V_{i,h}^{\epsilon_i}((\pi_{i,h}', \widehat{\pi}_{i,-h}) \times \widehat{\pi}_{-i} ; s).
\end{align}
Then, for any $i\in[n]$, the gap $V_{i,1}^{\epsilon_i}(\widehat{\pi}) - V_{i,1}^{\star,\epsilon_i}( \widehat{\pi}_{-i})$ can be decomposed as follows:
\begin{align}
    &V_{i,h}^{\epsilon_i}(\widehat{\pi}) - V_{i,h}^{\star, \epsilon_i}(\widehat{\pi}_{-i})\notag \\
    & = V_{i,h}^{\epsilon_i}(\widehat{\pi}) - \widehat{V}_{i,h}^{\epsilon_i}(\widehat{\pi})  + \left( \widehat{V}_{i,h}^{\epsilon_i}(\widehat{\pi}) - \widehat{V}_{i,h}^{\epsilon_i}((\pi_{i,h}^\star, \widehat{\pi}_{i,-h}) \times \widehat{\pi}_{-i}) \right) \notag \\
    &\quad + \left( \widehat{V}_{i,h}^{\epsilon_i}((\pi_{i,h}^\star, \widehat{\pi}_{i,-h}) \times \widehat{\pi}_{-i}) - V_{i,h}^{\epsilon_i}((\pi_{i,h}^\star, \widehat{\pi}_{i,-h}) \times \widehat{\pi}_{-i}) \right) \notag \\
    & \leq V_{i,h}^{\epsilon_i}(\widehat{\pi}) - \widehat{V}_{i,h}^{\epsilon_i}(\widehat{\pi})  + \left( \widehat{V}_{i,h}^{\epsilon_i}((\pi_{i,h}^\star, \widehat{\pi}_{i,-h}) \times \widehat{\pi}_{-i}) - V_{i,h}^{\epsilon_i}((\pi_{i,h}^\star, \widehat{\pi}_{i,-h}) \times \widehat{\pi}_{-i}) \right) \notag \\
    &\leq \|V_{i,h}^{\epsilon_i}(\widehat{\pi}) - \widehat{V}_{i,h}^{\epsilon_i}(\widehat{\pi})\|_\infty \bm{1}  + \left\|\widehat{V}_{i,h}^{\epsilon_i}((\pi_{i,h}^\star, \widehat{\pi}_{i,-h}) \times \widehat{\pi}_{-i}) - V_{i,h}^{\epsilon_i}((\pi_{i,h}^\star, \widehat{\pi}_{i,-h}) \times \widehat{\pi}_{-i})\right\|_\infty \bm{1} \label{eq:decomposition}
\end{align}
where the first inequality holds by applying Theorem~\ref{thm:nplayer-markov} with the estimated RAMG $\widehat{\mathcal{MG}}$ so that $\widehat{\pi}$ is a RQE of $\widehat{\mathcal{MG}}$, i.e.,
\begin{align}
     \widehat{V}_{i,h}^{\epsilon_i}(\widehat{\pi})  \leq  \min_{\pi_{i,h}':\cS  \mapsto \Delta_{A_i}} \widehat{V}_{i,h}^{\epsilon_i}((\pi_{i,h}', \widehat{\pi}_{i,-h}) \times \widehat{\pi}_{-i})  \leq \widehat{V}_{i,h}^{\epsilon_i}((\pi_{i,h}^\star, \widehat{\pi}_{i,-h}) \times \widehat{\pi}_{-i}). 
\end{align}

To continue, we divide the proof into several key steps.

\paragraph{Step 1: developing the recursion.}
To control the two terms in \eqref{eq:decomposition}, we consider that for any joint policy $\pi$ and time step $(s,h) \in\cS \times [H]$,
\begin{align}
    &V_{i,h}^{\epsilon_i}(\pi;s) - \widehat{V}_{i,h}^{\epsilon_i}(\pi;s) \notag \\
    & \overset{\mathrm{(i)}}{=}g_{\mathsf{pol},i}^{\pi}\left(Q_{i,h}(\pi; s,:) \right) - g_{\mathsf{pol},i}^{\pi}\left(\widehat{Q}_{i,h}(\pi;s,:) \right) \notag \\
    & \overset{\mathrm{(ii)}}{=} \sup_{p_i \in \cP_i} -\pi_i(s)^T g_{\mathsf{env},i}\left( R_{i,h}(s,:), P_{h,s,:}, V_{i,h+1}(\pi)\right) p_i -D_{\mathsf{pol},i}(p_i,\pi_{-i}(s))  \notag \\
    & \quad - \left[\sup_{p_i \in \cP_i }  -\pi_i(s)^T g_{\mathsf{env},i}\left( R_{i,h}(s,:), \widehat{P}_{h,s,:}, \widehat{V}_{i,h+1}(\pi)\right) p_i -D_{\mathsf{pol},i}(p_i,\pi_{-i}(s)) \right] \notag \\
    & \leq \sup_{p_i \in \cP_i} \left|-\pi_i(s)^T \left[g_{\mathsf{env},i}\left( R_{i,h}(s,:), P_{h,s,:}, V_{i,h+1}(\pi)\right)  - g_{\mathsf{env},i}\left( R_{i,h}(s,:), \widehat{P}_{h,s,:}, \widehat{V}_{i,h+1}(\pi)\right) \right] p_i  \right| \notag \\
    & \leq \left\|g_{\mathsf{env},i}\left( R_{i,h}(s,:), P_{h,s,:}, V_{i,h+1}(\pi)\right) - g_{\mathsf{env},i}\left( R_{i,h}(s,:), \widehat{P}_{h,s,:}, \widehat{V}_{i,h+1}(\pi)\right) \right\|_\infty, \label{eq:V-each-s}
\end{align}
where (i) holds by the definition in \eqref{eq:new-defined-value-epsilon} and \eqref{eq:RQEgame-markov-V}, and (ii) follows from the definition of $g_{\mathsf{pol},i}^{\pi}(\cdot)$ in \eqref{eq:RQEgame-markov}.

To continue, we know that for any $(s,\ba)\in\cS\times \cA$:
\begin{align}
&\left|g_{\mathsf{env},i}\left( R_{i,h}(s,\ba), P_{h,s,\ba}, V_{i,h+1}(\pi)\right) - g_{\mathsf{env},i}\left( R_{i,h}(s,:), \widehat{P}_{h,s,\ba}, \widehat{V}_{i,h+1}(\pi)\right) \right|\notag \\
& =  \Big| R_{i,h}(s, \ba)  - \sup_{\widetilde{P} \in \Delta_S } - \widetilde{P} V_{i,h+1}(\pi) - D_{\mathsf{env},i}(\widetilde{P}, P_{h,s,\ba}) \notag \\
& \quad - \left( R_{i,h}(s, \ba)  - \sup_{\widetilde{P} \in \Delta_S } - \widetilde{P} \widehat{V}_{i,h+1}(\pi) - D_{\mathsf{env},i}(\widetilde{P}, \widehat{P}_{h,s,\ba})\right) \Big| \notag \\
& =  \Big|   - \sup_{\widetilde{P} \in \Delta_S } \left(- \widetilde{P} V_{i,h+1}(\pi) - D_{\mathsf{env},i}(\widetilde{P}, P_{h,s,\ba}) \right) + \sup_{\widetilde{P} \in \Delta_S } \left(- \widetilde{P} \widehat{V}_{i,h+1}(\pi) - D_{\mathsf{env},i}(\widetilde{P}, \widehat{P}_{h,s,\ba}) \right) \Big| \notag \\
& \overset{\mathrm{(i)}}{\leq} \sup_{\widetilde{P} \in \Delta_S }  \Big|   - \left(- \widetilde{P} V_{i,h+1}(\pi) - D_{\mathsf{env},i}(\widetilde{P}, P_{h,s,\ba}) \right) + \left(- \widetilde{P} \widehat{V}_{i,h+1}(\pi) - D_{\mathsf{env},i}(\widetilde{P}, \widehat{P}_{h,s,\ba}) \right) \Big| \notag \\
& \overset{\mathrm{(ii)}}{=} \sup_{\widetilde{P} \in \Delta_S }  \Big|   \widetilde{P} \left( V_{i,h+1}(\pi) -\widehat{V}_{i,h+1}(\pi) \right) \Big | + \sup_{\widetilde{P} \in \Delta_S } \Big| D_{\mathsf{env},i}(\widetilde{P}, P_{h,s,\ba}) - D_{\mathsf{env},i}(\widetilde{P}, \widehat{P}_{h,s,\ba})  \Big| \notag \\
& \leq \left\| V_{i,h+1}(\pi) -  \widehat{V}_{i,h+1}(\pi) \right\|_\infty + L \left\| P_{h,s,\ba} -  \widehat{P}_{h,s,\ba} \right\|_1. \label{eq:recursive-sa}
\end{align}
where the first equality holds by  \eqref{eq:nvi-iteration}, (i) holds by the supreme operator is $1$-Lipschitz, and the last inequality holds by the assumption that $D_{\mathsf{env},i}(\cdot)$ is $L$-Lipschitz with respect to the $\ell_1$ norm for the second argument, with any fixed first argument.
We mention an important note here. As remarked earlier, similar sample complexity guarantees can be shown when $\{D_{\mathsf{env},i}\}_{i\in[n]}$ are defined as $\phi$-divergence and allude to analyses ideas in \cite{panaganti22a,xu-panaganti-2023samplecomplexity,shi2023curious} to modify this step.

Plugging in \eqref{eq:recursive-sa} back to \eqref{eq:V-each-s} and applying the results for all $s\in\cS$ yields 
\begin{align}
    \left\|V_{i,h}^{\epsilon_i}(\pi) - \widehat{V}_{i,h}^{ \epsilon_i}(\pi) \right\|_\infty &\leq \left\| V_{i,h+1}(\pi) -  \widehat{V}_{i,h+1}(\pi) \right\|_\infty + L \underbrace{\max_{(s,\ba)\in\cS \times \cA} \left\| P_{h,s,\ba} -  \widehat{P}_{h,s,\ba} \right\|_1 }_{=:l_{h}} \notag \\
    & = \left\| V^{\epsilon_i}_{i,h+1}(\pi) -  \widehat{V}^{\epsilon_i}_{i,h+1}(\pi) \right\|_\infty + L l_h.
\end{align}

Applying above fact recursively for $h,h+1,\cdots, H$, we arrive at
\begin{align}
    \left\|V_{i,h}^{\epsilon_i}(\pi) - \widehat{V}_{i,h}^{ \epsilon_i}(\pi) \right\|_\infty &\leq \left\| V_{i,h+2}^{\epsilon_i}(\pi) - \widehat{V}_{i,h+2}^{ \epsilon_i}(\pi) \right\|_\infty + Ll_h + Ll_{h+1} \notag \\
    & \leq \cdots \leq \left\| V^{\epsilon_i}_{i,H+1}(\pi) -  \widehat{V}^{\epsilon_i}_{i,H+1}(\pi) \right\|_\infty + L \sum_{t=h}^H  l_t \notag \\
    & \leq L \sum_{t=h}^H l_t, \label{eq:recursive-results}
\end{align}
where the last inequality holds since 
$V^{\epsilon_i}_{i,H+1}(\pi) = \widehat{V}^{\epsilon_i}_{i,H+1}(\pi) = V_{i,h+1}(\pi) = \widehat{V}_{i,h+1}(\pi) = 0$ for any policy $\pi$.

\paragraph{Step 2: controlling the errors $\{l_t\}$.} The remainder of the proof will focus on controlling \eqref{eq:recursive-results}.
Applying \cite[Lemma 17]{auer2008near} over all $(h,s,\ba)\in [H]\times \cS\times \cA$, we achieve the union bound that with probability at least $1-\delta$,
\begin{align}
\max_{(h,s,\ba)\in [H]\times \cS\times \cA} \left\| P_{h,s,\ba} -  \widehat{P}_{h,s,\ba} \right\|_1
& \leq \sqrt{ \frac{14S}{N}\log\left(\frac{2S\prod_{i\in[n]} A_iH}{\delta}\right)}. \label{eq:concentration-one}
\end{align}
Applying \eqref{eq:concentration-one}, we arrive at with probability at least $1-\delta$,
\begin{align}
   \forall h\in[H],\, l_h = \max_{(s,\ba)\in\cS \times \cA} \left\| P_{h,s,\ba} -  \widehat{P}_{h,s,\ba} \right\|_1 \leq \sqrt{ \frac{14S}{N}\log\left(\frac{2S\prod_{i\in[n]} A_iH}{\delta}\right)}. \label{eq:union-bound}
\end{align}

Inserting \eqref{eq:union-bound} into \eqref{eq:recursive-results} gives
\begin{align}
    \left\|V_{i,h}^{\epsilon_i}(\pi) - \widehat{V}_{i,h}^{ \epsilon_i}(\pi)  \right\|_\infty &\leq L\sum_{t=h}^H l_t \leq HL\sqrt{ \frac{14S}{N}\log\left(\frac{2S\prod_{i\in[n]} A_iH}{\delta}\right)}. \label{eq:one-term}
\end{align}

Finally, recalling \eqref{eq:decomposition} yields

\begin{align}
    V_{i,1}^{ \epsilon_i}(\widehat{\pi}) - V_{i,1}^{\star, \epsilon_i}(\widehat{\pi}_{-i})
    &\leq\|V_{i,h}^{\epsilon_i}(\widehat{\pi}) - \widehat{V}_{i,h}^{\epsilon_i}(\widehat{\pi})\|_\infty \bm{1}  + \left\|\widehat{V}_{i,h}^{\epsilon_i}((\pi_{i,h}^\star, \widehat{\pi}_{i,-h}) \times \widehat{\pi}_{-i}) - V_{i,h}^{\epsilon_i}((\pi_{i,h}^\star, \widehat{\pi}_{i,-h}) \times \widehat{\pi}_{-i})\right\|_\infty \bm{1} \notag \\
    & \leq 8HL\sqrt{ \frac{S}{N}\log\left(\frac{2S\prod_{i\in[n]} A_iH}{\delta}\right)} \bm{1},
\end{align}
where the last inequality holds by applying \eqref{eq:one-term} with $\pi = \widehat{\pi}$ or $\pi = (\pi_{i,h}^\star, \widehat{\pi}_{i,-h}) \times \widehat{\pi}_{-i}$.
\hfill$ \square$

\section{Experiment Details} \label{appen:exp}

We provide more results and details of our experiments in this section.

\subsection{Matrix Games} \label{appen:exp:matching-pennies}


\begin{table}[ht]
\centering
\begin{tabular}{|c|c|c|}
\hline
\textbf{Game 1} & \textbf{Game 2} & \textbf{Game 3} \\ \hline
\begin{tabular}{c|c|c}
 & L & R \\ \hline
U & $\begin{matrix}10 \\ 10\end{matrix}$ & $\begin{matrix}0 \\ 18\end{matrix}$ \\ \hline
D & $\begin{matrix}9 \\ 9\end{matrix}$ & $\begin{matrix}10 \\ 8\end{matrix}$ \\ 
\end{tabular} &
\begin{tabular}{c|c|c}
 & L & R \\ \hline
U & $\begin{matrix}9 \\ 4\end{matrix}$ & $\begin{matrix}0 \\ 13\end{matrix}$ \\ \hline
D & $\begin{matrix}6 \\ 7\end{matrix}$ & $\begin{matrix}8 \\ 5\end{matrix}$ \\ 
\end{tabular} &
\begin{tabular}{c|c|c}
 & L & R \\ \hline
U & $\begin{matrix}8 \\ 6\end{matrix}$ & $\begin{matrix}0 \\ 14\end{matrix}$ \\ \hline
D & $\begin{matrix}7 \\ 7\end{matrix}$ & $\begin{matrix}10 \\ 4\end{matrix}$ \\ 
\end{tabular} \\ \hline
\textbf{Game 4} & \textbf{Game 5} & \textbf{Game 6} \\ \hline
\begin{tabular}{c|c|c}
 & L & R \\ \hline
U & $\begin{matrix}7 \\ 4\end{matrix}$ & $\begin{matrix}0 \\ 11\end{matrix}$ \\ \hline
D & $\begin{matrix}5 \\ 6\end{matrix}$ & $\begin{matrix}9 \\ 2\end{matrix}$ \\ 
\end{tabular} &
\begin{tabular}{c|c|c}
 & L & R \\ \hline
U & $\begin{matrix}7 \\ 2\end{matrix}$ & $\begin{matrix}0 \\ 9\end{matrix}$ \\ \hline
D & $\begin{matrix}4 \\ 5\end{matrix}$ & $\begin{matrix}8 \\ 1\end{matrix}$ \\ 
\end{tabular} &
\begin{tabular}{c|c|c}
 & L & R \\ \hline
U & $\begin{matrix}7 \\ 1\end{matrix}$ & $\begin{matrix}1 \\ 7\end{matrix}$ \\ \hline
D & $\begin{matrix}3 \\ 5\end{matrix}$ & $\begin{matrix}8 \\ 0\end{matrix}$ \\ 
\end{tabular} \\ \hline
\textbf{Game 7} & \textbf{Game 8} & \textbf{Game 9} \\ \hline
\begin{tabular}{c|c|c}
 & L & R \\ \hline
U & $\begin{matrix}10 \\ 12\end{matrix}$ & $\begin{matrix}4 \\ 22\end{matrix}$ \\ \hline
D & $\begin{matrix}9 \\ 9\end{matrix}$ & $\begin{matrix}14 \\ 8\end{matrix}$ \\ 
\end{tabular} &
\begin{tabular}{c|c|c}
 & L & R \\ \hline
U & $\begin{matrix}9 \\ 7\end{matrix}$ & $\begin{matrix}3 \\ 16\end{matrix}$ \\ \hline
D & $\begin{matrix}6 \\ 7\end{matrix}$ & $\begin{matrix}11 \\ 5\end{matrix}$ \\ 
\end{tabular} &
\begin{tabular}{c|c|c}
 & L & R \\ \hline
U & $\begin{matrix}8 \\ 9\end{matrix}$ & $\begin{matrix}3 \\ 17\end{matrix}$ \\ \hline
D & $\begin{matrix}7 \\ 7\end{matrix}$ & $\begin{matrix}13 \\ 4\end{matrix}$ \\ 
\end{tabular} \\ \hline
\textbf{Game 10} & \textbf{Game 11} & \textbf{Game 12} \\ \hline
\begin{tabular}{c|c|c}
 & L & R \\ \hline
U & $\begin{matrix}7 \\ 6\end{matrix}$ & $\begin{matrix}2 \\ 13\end{matrix}$ \\ \hline
D & $\begin{matrix}5 \\ 6\end{matrix}$ & $\begin{matrix}11 \\ 2\end{matrix}$ \\ 
\end{tabular} &
\begin{tabular}{c|c|c}
 & L & R \\ \hline
U & $\begin{matrix}7 \\ 4\end{matrix}$ & $\begin{matrix}2 \\ 11\end{matrix}$ \\ \hline
D & $\begin{matrix}4 \\ 5\end{matrix}$ & $\begin{matrix}10 \\ 1\end{matrix}$ \\ 
\end{tabular} &
\begin{tabular}{c|c|c}
 & L & R \\ \hline
U & $\begin{matrix}7 \\ 3\end{matrix}$ & $\begin{matrix}3 \\ 9\end{matrix}$ \\ \hline
D & $\begin{matrix}3 \\ 5\end{matrix}$ & $\begin{matrix}10 \\ 0\end{matrix}$ \\ 
\end{tabular} \\ \hline
\end{tabular}
\caption{Payoff Matrices from \cite{selten2008stationary}. In each column, the number above (below) is the payoff for player 1 (player 2).}
\label{tab:payoff_matrices-SC}
\end{table}

\begin{table}[h!]
    \centering
    \begin{tabular}{c|c|c}
 & L & R \\ \hline
U & $\begin{matrix}200 \\ 160\end{matrix}$ & $\begin{matrix}160 \\ 10\end{matrix}$ \\ \hline
D & $\begin{matrix}370 \\ 200\end{matrix}$ & $\begin{matrix}10 \\ 370\end{matrix}$ \\ 
\end{tabular}
    \caption{Game 4 Payoff Matrix from \cite{goeree2003risk}. In each column, the number above (below) is the payoff for player 1 (player 2).}
    \label{tab:payoff_matrices-GHP}
\end{table}

\paragraph{Matching pennies matrix games:} Two players in matching pennies simultaneously choose heads or tails, and a player wins (the other player loses) a payoff if their choices \textit{match}. The detailed payoff matrices are provided in \cref{tab:payoff_matrices-GHP,tab:payoff_matrices-SC}. These games from \cite{goeree2003risk,selten2008stationary} are focused on showcasing risk-averse solutions by developing payoffs strategically such that either player deviating from their choice of plays will cause hefty damage in terms of payoff to the other player. So, both players prefer a `safer' choice of plays, thus highlighting both risk-averse and bounded rational preferences. 

\vspace{-0.5cm}
\paragraph{Algorithm and Result details:}
We use the penalty functions KL and reverse KL with the regularizers log barrier and negative entropy, as mentioned in \cref{tab:risk-measure}, for our experiments. To be precise, we consider the following two experimental setups: 
(a) For $j=1,2$, we let $\nu_j(p)=-\sum_i\log(p_i)$ and $D_j(p,q)=(1/\tau_j)\mathrm{KL}(p,q)$. 
(b) For $j=1,2$, we let $\nu_j(p)=\sum_i p_i\log(p_i)$ and $D_j(p,q)=(1/\tau_j)\mathrm{KL}(q,p)$. We note that the parameters $\tau_j^{-1}$ play the same role as $\xi_j$'s (formalized in \cref{cor:KL,cor:reverse-KL}), i.e., the players take more risk neutral decisions as $\tau_j\to0$. 
With the perfect information, we use the vanilla projected gradient descent with constant stepsizes \cite{beck2017first} as the no-regret algorithm to arrive at the CCE $\sigma$ of the four player game. 
We use consistent stepsizes between $10^{-4}$ and $10^{-3}$ for all our runs iterating through $10^4$ steps of gradient descent. We notice $2\%$ deviations in sup-norm for the algorithm policies in about 20 runs.

\subsection{Markov Games} \label{appen:cliff-l1}

\begin{figure}[ht]
	\centering
    \vspace{-0.5cm}
	\includegraphics[width=1.0\linewidth]{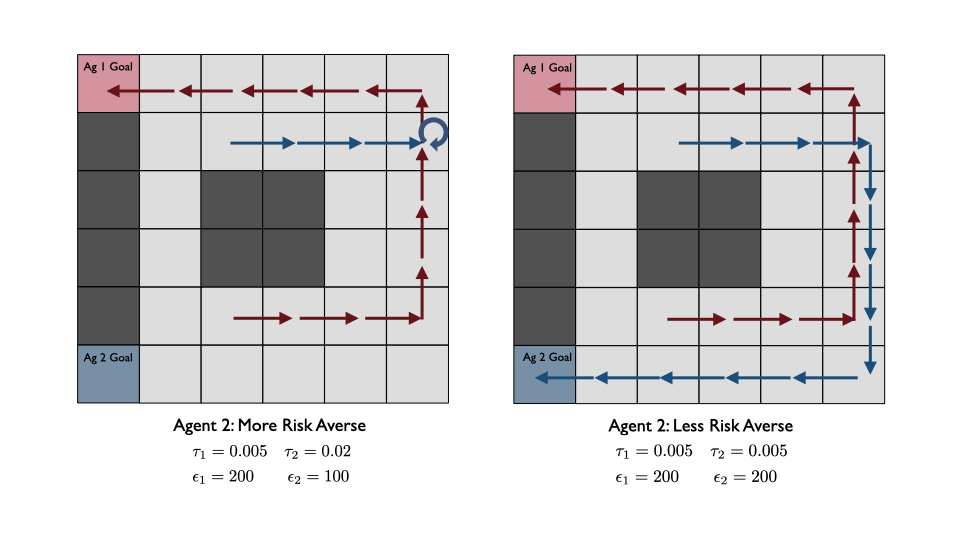}
    \vspace{-1cm}
	\caption{ Cliff-Walk results for the $\ell_1$ environmental uncertainty metric.
 }
	\label{fig:cliff-l1}
\end{figure}

Here, we showcase results corresponding to $\ell_1$ metric for the environmental uncertainty metric $D_{\mathsf{env},i}$ for the grid-world problem.
We consider the same \emph{Cliff-Walk} environment described in \cref{appen:exp:grid-world} under some minor modifications discussed here. The rewards for falling into the cliff is now $-100$ and for reaching the goal is $20$.
Agents/players are rewarded a small negative reward of $-0.1$ for taking each step.
The episode horizon $H$ is $100$.
The algorithm details is the same as described in \cref{appen:exp:grid-world}.

We present two results in \cref{fig:cliff-l1} that has similar implications as in \cref{fig:cliff}. We make one important observation about these agents equipped with $\ell_1$ environmental uncertainty metric. These agents are less risk-averse towards the cliff in horizontal grid axis compared to the results \cref{fig:cliff} corresponding to the KL environmental uncertainty metric. We do not investigate the effects of different environmental uncertainty metrics in this work and postpone to address in future research directions.

%% file: arxiv-risk-averse-game.bbl
\begin{thebibliography}{10}

\bibitem{agarwal2020model}
Alekh Agarwal, Sham Kakade, and Lin~F Yang.
\newblock Model-based reinforcement learning with a generative model is minimax optimal.
\newblock In {\em Conference on Learning Theory}, pages 67--83. PMLR, 2020.

\bibitem{ahmadi2012entropic}
Amir Ahmadi-Javid.
\newblock Entropic value-at-risk: A new coherent risk measure.
\newblock {\em Journal of Optimization Theory and Applications}, 155:1105--1123, 2012.

\bibitem{auer2008near}
Peter Auer, Thomas Jaksch, and Ronald Ortner.
\newblock Near-optimal regret bounds for reinforcement learning.
\newblock {\em Advances in neural information processing systems}, 21, 2008.

\bibitem{AUMANN197467}
Robert~J. Aumann.
\newblock Subjectivity and correlation in randomized strategies.
\newblock {\em Journal of Mathematical Economics}, 1(1):67--96, 1974.

\bibitem{CE2}
Robert~J. Aumann.
\newblock Correlated equilibrium as an expression of bayesian rationality.
\newblock {\em Econometrica}, 55(1):1--18, 1987.

\bibitem{beck2017first}
Amir Beck.
\newblock {\em First-order methods in optimization}.
\newblock SIAM, 2017.

\bibitem{blanchet2024double}
Jose Blanchet, Miao Lu, Tong Zhang, and Han Zhong.
\newblock Double pessimism is provably efficient for distributionally robust offline reinforcement learning: Generic algorithm and robust partial coverage.
\newblock {\em Advances in Neural Information Processing Systems}, 36, 2024.

\bibitem{BorkarRisk}
Vivek Borkar.
\newblock Risk-sensitive control, single controller games and linear programming.
\newblock {\em Journal of Dynamics and Games}, 11, 01 2023.

\bibitem{Experiments1}
James~N Brown and Robert Rosenthal.
\newblock Testing the minimax hypothesis: A re-examination of o'neill's game experiment.
\newblock {\em Econometrica}, 58(5):1065--81, 1990.

\bibitem{eq_collapse1}
Yang Cai, Ozan Candogan, Constantinos Daskalakis, and Christos Papadimitriou.
\newblock Zero-sum polymatrix games: A generalization of minmax.
\newblock {\em Mathematics of Operations Research}, 41(2):648--655, 2016.

\bibitem{noisy_decisions}
C.~Monica Capra, Jacob~K. Goeree, Rosario Gomez, and Charles~A. Holt.
\newblock Learning and noisy equilibrium behavior in an experimental study of imperfect price competition.
\newblock {\em International Economic Review}, 43(3):613--636, 2002.

\bibitem{cen2021fast}
Shicong Cen, Yuting Wei, and Yuejie Chi.
\newblock Fast policy extragradient methods for competitive games with entropy regularization.
\newblock {\em Advances in Neural Information Processing Systems}, 34:27952--27964, 2021.

\bibitem{bianchi_prediction}
Nicolo Cesa-Bianchi and Gabor Lugosi.
\newblock {\em Prediction, learning, and games.}
\newblock Cambridge University Press, 2006.

\bibitem{daskalakis2013complexity}
Constantinos Daskalakis.
\newblock On the complexity of approximating a nash equilibrium.
\newblock {\em ACM Transactions on Algorithms (TALG)}, 9(3):1--35, 2013.

\bibitem{daskalakis2023smooth}
Constantinos Daskalakis, Noah Golowich, Nika Haghtalab, and Abhishek Shetty.
\newblock Smooth nash equilibria: Algorithms and complexity.
\newblock {\em arXiv preprint arXiv:2309.12226}, 2023.

\bibitem{MarkovComplex}
Constantinos Daskalakis, Noah Golowich, and Kaiqing Zhang.
\newblock The complexity of markov equilibrium in stochastic games.
\newblock In Gergely Neu and Lorenzo Rosasco, editors, {\em Proceedings of Thirty Sixth Conference on Learning Theory}, volume 195 of {\em Proceedings of Machine Learning Research}, pages 4180--4234. PMLR, 12--15 Jul 2023.

\bibitem{DEKEL1990243}
Eddie Dekel and Drew Fudenberg.
\newblock Rational behavior with payoff uncertainty.
\newblock {\em Journal of Economic Theory}, 52(2):243--267, 1990.

\bibitem{prediction1}
Ido Erev and Alvin Roth.
\newblock Predicting how people play games: Reinforcement learning in experimental games with unique, mixed strategy equilibria.
\newblock {\em American Economic Review}, 88(4):848--81, 1998.

\bibitem{eriksson2022risk}
Hannes Eriksson, Debabrota Basu, Mina Alibeigi, and Christos Dimitrakakis.
\newblock Risk-sensitive bayesian games for multi-agent reinforcement learning under policy uncertainty.
\newblock {\em arXiv preprint arXiv:2203.10045}, 2022.

\bibitem{evans2024learning}
Benjamin~Patrick Evans and Sumitra Ganesh.
\newblock Learning and calibrating heterogeneous bounded rational market behaviour with multi-agent reinforcement learning.
\newblock {\em arXiv preprint arXiv:2402.00787}, 2024.

\bibitem{even2009convergence}
Eyal Even-Dar, Yishay Mansour, and Uri Nadav.
\newblock On the convergence of regret minimization dynamics in concave games.
\newblock In {\em Proceedings of the forty-first annual ACM symposium on Theory of computing}, pages 523--532, 2009.

\bibitem{fiat2010players}
Amos Fiat and Christos Papadimitriou.
\newblock When the players are not expectation maximizers.
\newblock In {\em Algorithmic Game Theory: Third International Symposium, SAGT 2010, Athens, Greece, October 18-20, 2010. Proceedings 3}, pages 1--14. Springer, 2010.

\bibitem{Risk_overview}
Hans F{\"o}llmer and Alexander Schied.
\newblock Convex measures of risk and trading constraints.
\newblock {\em Finance and Stochastics}, 6(4):429--447, 2002.

\bibitem{ganesh2019reinforcement}
Sumitra Ganesh, Nelson Vadori, Mengda Xu, Hua Zheng, Prashant Reddy, and Manuela Veloso.
\newblock Reinforcement learning for market making in a multi-agent dealer market.
\newblock {\em Advances in Neural Information Processing Systems}, 2019.

\bibitem{gao2021robust}
Yue Gao, Kry Yik~Chau Lui, and Pablo Hernandez-Leal.
\newblock Robust risk-sensitive reinforcement learning agents for trading markets.
\newblock {\em arXiv preprint arXiv:2107.08083}, 2021.

\bibitem{Gemp2017OnlineMG}
Ian~M. Gemp and Sridhar Mahadevan.
\newblock Online monotone games.
\newblock {\em ArXiv}, abs/1710.07328, 2017.

\bibitem{StochasticGameTheory}
Jacob~K. Goeree and Charles~A. Holt.
\newblock Stochastic game theory: For playing games, not just for doing theory.
\newblock {\em Proceedings of the National Academy of Sciences}, 96(19):10564--10567, 1999.

\bibitem{goeree2003risk}
Jacob~K Goeree, Charles~A Holt, and Thomas~R Palfrey.
\newblock Risk averse behavior in generalized matching pennies games.
\newblock {\em Games and Economic Behavior}, 45(1):97--113, 2003.

\bibitem{GoereeAuction}
K.~Goeree and Theo Offerman.
\newblock Efficiency in auctions with private and common values: An experimental study.
\newblock {\em American Economic Review}, 92(3):625--643, June 2002.

\bibitem{gollier2001economics}
Christian Gollier.
\newblock {\em The economics of risk and time}.
\newblock MIT press, 2001.

\bibitem{golowich2020tight}
Noah Golowich, Sarath Pattathil, and Constantinos Daskalakis.
\newblock Tight last-iterate convergence rates for no-regret learning in multi-player games.
\newblock {\em Advances in neural information processing systems}, 33:20766--20778, 2020.

\bibitem{he2023robust}
Sihong He, Songyang Han, Sanbao Su, Shuo Han, Shaofeng Zou, and Fei Miao.
\newblock Robust multi-agent reinforcement learning with state uncertainty.
\newblock {\em Transactions on Machine Learning Research}, 2023.

\bibitem{Camerer}
Teck Ho, Colin Camerer, and Juin-Kuan Chong.
\newblock A cognitive hierarchy model games.
\newblock {\em The Quarterly Journal of Economics}, 119:861--898, 02 2004.

\bibitem{iyengar2005robust}
Garud~N Iyengar.
\newblock Robust dynamic programming.
\newblock {\em Mathematics of Operations Research}, 30(2):257--280, 2005.

\bibitem{jacob2022modeling}
Athul~Paul Jacob, David~J Wu, Gabriele Farina, Adam Lerer, Hengyuan Hu, Anton Bakhtin, Jacob Andreas, and Noam Brown.
\newblock Modeling strong and human-like gameplay with kl-regularized search.
\newblock In {\em International Conference on Machine Learning}, pages 9695--9728. PMLR, 2022.

\bibitem{Jacobsen}
D.~Jacobson.
\newblock Optimal stochastic linear systems with exponential performance criteria and their relation to deterministic differential games.
\newblock {\em IEEE Transactions on Automatic Control}, 18(2):124--131, 1973.

\bibitem{Kaheman}
Daniel Kahneman and Amos Tversky.
\newblock Prospect theory: An analysis of decision under risk.
\newblock {\em Econometrica}, 47(2):263--291, 1979.

\bibitem{eq_collapse_Markov}
Fivos Kalogiannis and Ioannis Panageas.
\newblock Zero-sum polymatrix markov games: Equilibrium collapse and efficient computation of nash equilibria, 2023.

\bibitem{kannan2017fairness}
Sampath Kannan, Michael Kearns, Jamie Morgenstern, Mallesh Pai, Aaron Roth, Rakesh Vohra, and Zhiwei~Steven Wu.
\newblock Fairness incentives for myopic agents.
\newblock In {\em Proceedings of the 2017 ACM Conference on Economics and Computation}, pages 369--386, 2017.

\bibitem{kearns1998finite}
Michael Kearns and Satinder Singh.
\newblock Finite-sample convergence rates for q-learning and indirect algorithms.
\newblock {\em Advances in neural information processing systems}, 11, 1998.

\bibitem{quantal3}
Stefanos Leonardos, Georgios Piliouras, and Kelly Spendlove.
\newblock Exploration-exploitation in multi-agent competition: Convergence with bounded rationality.
\newblock In M.~Ranzato, A.~Beygelzimer, Y.~Dauphin, P.S. Liang, and J.~Wortman Vaughan, editors, {\em Advances in Neural Information Processing Systems}, volume~34, pages 26318--26331. Curran Associates, Inc., 2021.

\bibitem{Luce59}
R.~Duncan Luce.
\newblock {\em Individual Choice Behavior: A Theoretical analysis}.
\newblock Wiley, New York, NY, USA, 1959.

\bibitem{Experiment3}
Richard~D. McKelvey and Thomas~R. Palfrey.
\newblock An experimental study of the centipede game.
\newblock {\em Econometrica}, 60(4):803--836, 1992.

\bibitem{mckelvey1995quantal}
Richard~D McKelvey and Thomas~R Palfrey.
\newblock Quantal response equilibria for normal form games.
\newblock {\em Games and economic behavior}, 10(1):6--38, 1995.

\bibitem{mckelvey1998quantal}
Richard~D McKelvey and Thomas~R Palfrey.
\newblock Quantal response equilibria for extensive form games.
\newblock {\em Experimental economics}, 1:9--41, 1998.

\bibitem{mcmahan2024roping}
Jeremy McMahan, Giovanni Artiglio, and Qiaomin Xie.
\newblock Roping in uncertainty: Robustness and regularization in markov games.
\newblock In {\em Forty-first International Conference on Machine Learning}, 2024.

\bibitem{quantal2}
Panayotis Mertikopoulos and William~H. Sandholm.
\newblock Learning in games via reinforcement and regularization.
\newblock {\em Mathematics of Operations Research}, 41(4):1297--1324, 2016.

\bibitem{mohsenian2010autonomous}
Amir-Hamed Mohsenian-Rad, Vincent~WS Wong, Juri Jatskevich, Robert Schober, and Alberto Leon-Garcia.
\newblock Autonomous demand-side management based on game-theoretic energy consumption scheduling for the future smart grid.
\newblock {\em IEEE transactions on Smart Grid}, 1(3):320--331, 2010.

\bibitem{CCE1}
H.~Moulin and J.~P. Vial.
\newblock Strategically zero-sum games: The class of games whose completely mixed equilibria cannot be improved upon.
\newblock {\em International Journal of Game Theory}, 7(3):201--221, 1978.

\bibitem{munos2023nash}
Remi Munos, Michal Valko, Daniele Calandriello, Mohammad~Gheshlaghi Azar, Mark Rowland, Zhaohan~Daniel Guo, Yunhao Tang, Matthieu Geist, Thomas Mesnard, Andrea Michi, Marco Selvi, Sertan Girgin, Nikola Momchev, Olivier Bachem, Daniel~J. Mankowitz, Doina Precup, and Bilal Piot.
\newblock Nash learning from human feedback, 2023.

\bibitem{nash1950non}
John~F Nash.
\newblock Non-cooperative games.
\newblock {\em Princeton University}, 1950.

\bibitem{Experiments2}
Barry O'Neill.
\newblock Nonmetric test of the minimax theory of two-person zerosum games.
\newblock {\em Proceedings of the National Academy of Sciences of the United States of America}, 84(7):2106--2109, 1987.

\bibitem{panaganti2020robust}
Kishan Panaganti and Dileep Kalathil.
\newblock Robust reinforcement learning using least squares policy iteration with provable performance guarantees.
\newblock In {\em International Conference on Machine Learning (ICML)}, pages 511--520, 2021.

\bibitem{panaganti22a}
Kishan Panaganti and Dileep Kalathil.
\newblock Sample complexity of robust reinforcement learning with a generative model.
\newblock In {\em International Conference on Artificial Intelligence and Statistics (AISTATS)}, pages 9582--9602, 2022.

\bibitem{qiu2021rmix}
Wei Qiu, Xinrun Wang, Runsheng Yu, Rundong Wang, Xu~He, Bo~An, Svetlana Obraztsova, and Zinovi Rabinovich.
\newblock Rmix: Learning risk-sensitive policies for cooperative reinforcement learning agents.
\newblock {\em Advances in Neural Information Processing Systems}, 34:23049--23062, 2021.

\bibitem{Rosen}
J.~B. Rosen.
\newblock Existence and uniqueness of equilibrium points for concave n-person games.
\newblock {\em Econometrica}, 33(3):520--534, 1965.

\bibitem{roughgarden2015intrinsic}
Tim Roughgarden.
\newblock Intrinsic robustness of the price of anarchy.
\newblock {\em Journal of the ACM (JACM)}, 62(5):1--42, 2015.

\bibitem{selfridge1989adaptive}
Oliver~G Selfridge.
\newblock Adaptive strategies of learning a study of two-person zero-sum competition.
\newblock In {\em Proceedings of the sixth international workshop on Machine learning}, pages 412--415. Elsevier, 1989.

\bibitem{selten2008stationary}
Reinhard Selten and Thorsten Chmura.
\newblock Stationary concepts for experimental 2x2-games.
\newblock {\em American Economic Review}, 98(3):938--966, 2008.

\bibitem{shapley1953stochastic}
Lloyd~S Shapley.
\newblock Stochastic games.
\newblock {\em Proceedings of the national academy of sciences}, 39(10):1095--1100, 1953.

\bibitem{shen2023riskq}
Siqi Shen, Chennan Ma, Chao Li, Weiquan Liu, Yongquan Fu, Songzhu Mei, Xinwang Liu, and Cheng Wang.
\newblock Riskq: risk-sensitive multi-agent reinforcement learning value factorization.
\newblock {\em Advances in Neural Information Processing Systems}, 36:34791--34825, 2023.

\bibitem{shenControl}
Yun Shen, Wilhelm Stannat, and Klaus Obermayer.
\newblock Risk-sensitive markov control processes.
\newblock {\em SIAM Journal on Control and Optimization}, 51(5):3652--3672, 2013.

\bibitem{shen2014risk}
Yun Shen, Michael~J Tobia, Tobias Sommer, and Klaus Obermayer.
\newblock Risk-sensitive reinforcement learning.
\newblock {\em Neural computation}, 26(7):1298--1328, 2014.

\bibitem{shi2022distributionally}
Laixi Shi and Yuejie Chi.
\newblock Distributionally robust model-based offline reinforcement learning with near-optimal sample complexity.
\newblock {\em arXiv preprint arXiv:2208.05767}, 2022.

\bibitem{shi2023curious}
Laixi Shi, Gen Li, Yuting Wei, Yuxin Chen, Matthieu Geist, and Yuejie Chi.
\newblock The curious price of distributional robustness in reinforcement learning with a generative model.
\newblock {\em Advances in Neural Information Processing Systems}, 2023.

\bibitem{shi2024sample}
Laixi Shi, Eric Mazumdar, Yuejie Chi, and Adam Wierman.
\newblock Sample-efficient robust multi-agent reinforcement learning in the face of environmental uncertainty.
\newblock {\em arXiv preprint arXiv:2404.18909}, 2024.

\bibitem{slumbers2023game}
Oliver Slumbers, David~Henry Mguni, Stefano~B Blumberg, Stephen~Marcus Mcaleer, Yaodong Yang, and Jun Wang.
\newblock A game-theoretic framework for managing risk in multi-agent systems.
\newblock In {\em International Conference on Machine Learning}, pages 32059--32087. PMLR, 2023.

\bibitem{quantal1}
Samuel Sokota, Ryan D'Orazio, J~Zico Kolter, Nicolas Loizou, Marc Lanctot, Ioannis Mitliagkas, Noam Brown, and Christian Kroer.
\newblock A unified approach to reinforcement learning, quantal response equilibria, and two-player zero-sum games.
\newblock In {\em The Eleventh International Conference on Learning Representations}, 2023.

\bibitem{Tversky1992}
Amos Tversky and Daniel Kahneman.
\newblock Advances in prospect theory: Cumulative representation of uncertainty.
\newblock {\em Journal of Risk and Uncertainty}, 5:297--323, 1992.

\bibitem{no_regret_dominated}
Yannick Viossat and Andriy Zapechelnyuk.
\newblock No-regret dynamics and fictitious play.
\newblock {\em Journal of Economic Theory}, 148(2):825--842, 2013.

\bibitem{wang2024learning}
Zifan Wang, Yi~Shen, Michael~M Zavlanos, and Karl~H Johansson.
\newblock Learning of nash equilibria in risk-averse games.
\newblock {\em arXiv preprint arXiv:2403.10399}, 2024.

\bibitem{wellman1998market}
Michael~P Wellman and Peter~R Wurman.
\newblock Market-aware agents for a multiagent world.
\newblock {\em Robotics and Autonomous Systems}, 24(3-4):115--125, 1998.

\bibitem{xu-panaganti-2023samplecomplexity}
Zaiyan Xu$^{*}$, Kishan Panaganti$^{*}$, and Dileep Kalathil.
\newblock Improved sample complexity bounds for distributionally robust reinforcement learning.
\newblock In {\em Proceedings of The 25th International Conference on Artificial Intelligence and Statistics}. Conference on Artificial Intelligence and Statistics, 2023.

\bibitem{yekkehkhany2020risk}
Ali Yekkehkhany, Timothy Murray, and Rakesh Nagi.
\newblock Risk-averse equilibrium for games.
\newblock {\em arXiv preprint arXiv:2002.08414}, 2020.

\bibitem{zhang2020robust}
Kaiqing Zhang, Tao Sun, Yunzhe Tao, Sahika Genc, Sunil Mallya, and Tamer Basar.
\newblock Robust multi-agent reinforcement learning with model uncertainty.
\newblock {\em Advances in neural information processing systems}, 33:10571--10583, 2020.

\bibitem{zhang2021multi}
Kaiqing Zhang, Zhuoran Yang, and Tamer Ba{\c{s}}ar.
\newblock Multi-agent reinforcement learning: A selective overview of theories and algorithms.
\newblock {\em Handbook of reinforcement learning and control}, pages 321--384, 2021.

\bibitem{zhang2024soft}
Runyu Zhang, Yang Hu, and Na~Li.
\newblock Soft robust {MDP}s and risk-sensitive {MDP}s: Equivalence, policy gradient, and sample complexity.
\newblock In {\em The Twelfth International Conference on Learning Representations}, 2024.

\bibitem{zhang2021mean}
Shangtong Zhang, Bo~Liu, and Shimon Whiteson.
\newblock Mean-variance policy iteration for risk-averse reinforcement learning.
\newblock In {\em Proceedings of the AAAI Conference on Artificial Intelligence}, volume~35, pages 10905--10913, 2021.

\end{thebibliography}
